\documentclass[journal,draftcls,onecolumn,12pt,twoside]{IEEEtran}
\usepackage{cite}
\usepackage{graphicx}
\usepackage{psfrag}
\usepackage{subfigure}
\usepackage{url}
\usepackage{amsmath}
\usepackage[normalem]{ulem}
\usepackage{epsfig,colortbl}
\usepackage{amssymb,comment}
\usepackage{enumerate}
\usepackage{times}
\usepackage{multirow,multicol}
\usepackage[ruled, vlined]{algorithm2e}
\usepackage{tabularx}

\newcounter{ctr}

\newtheorem{lemma}{Lemma}
\newtheorem{theorem}{Theorem}

\newtheorem{corollary}{Corollary}

\makeatletter

\newcommand{\be}{\begin{eqnarray}}
\newcommand{\ee}{\end{eqnarray}}
\newcommand{\nn}{\nonumber}

\newcommand{\bm}{\boldmath}

\newcommand{\mc}{\multicolumn}

\newcommand{\m}{\mbox{\bm $m$}}

\newcommand{\uu}{\mbox{\bm $u$}}
\newcommand{\ccc}{{\mbox{\bm $c$}}}

\newcommand{\eee}{\mbox{\bm $e$}}

\newcommand{\gggg}{\mbox{\bm $g$}}

\newcommand{\0}{\mbox{\bm  $0$}}

\newcommand{\zz}{\mbox{\bm $z$}}

\renewcommand\paragraph{\@startsection{paragraph}{4}{\z@}%
    {1.5ex plus .2ex minus .3ex}%
            {-0em}%
                        {\normalsize\bf}}

\begin{document}
\title{Update-Efficient Error-Correcting Product-Matrix Codes }
\author{Yunghsiang S. Han\thanks{Part of this work was presented at the IEEE
International Symposium on Information Theory (ISIT~2013).
This work was supported in part by  CASE: The Center for Advanced Systems and Engineering, a NYSTAR center for advanced technology at Syracuse University; the National Science of Council (NSC) of
Taiwan under grants no. 99-2221-E-011-158-MY3 and NSC 101-2221-E-011-069-MY3;
US National Science Foundation under grant no. CNS-1117560 and McMaster
University new faculty startup fund. Han's work was completed during his visit of Syracuse University from 2012 to 2013.

Han is with the Dept. of
          Electrical Engineering, National Taiwan University of Science and Technology,
          Taipei, Taiwan
          (e-mail: yshan\mbox{@}mail.ntust.edu.tw), Pai is with the Graduate Institute of communication  Engineering, National Taipei University, Taiwan, R.O.C., Zheng is with the Dept. of Computing and Software, McMaster University, Hamilton, ON, Canada, and  Varshney is with the  Dept. of Electrical Engineering and Computer  Science, Syracuse University, Syracuse, NY USA.}, {\it Fellow}, {\it IEEE}, Hung-Ta~Pai, {\it Senior Member}. {\it IEEE}, Rong~Zheng, {\it Senior Member}. {\it IEEE}, Pramod~K. Varshney,  {\it Fellow}, {\it IEEE}} \pubid{}
\maketitle

\begin{abstract}
Regenerating codes provide an efficient way to recover data at failed nodes in
distributed storage systems. It has been shown that regenerating codes can be designed
to minimize the per-node storage (called MSR) or minimize the communication
overhead for regeneration (called MBR).  In this work, we propose new
encoding schemes for $[n,d]$ error-correcting MSR and MBR codes that generalize our
earlier work on error-correcting regenerating codes.  We show that by choosing
a suitable diagonal matrix, any generator matrix of the $[n,\alpha]$
Reed-Solomon (RS) code can be integrated into the encoding matrix.  Hence, MSR
codes with the least update complexity can be found.  By using the coefficients of generator polynomials of $[n,k]$ and $[n,d]$ RS codes, we present a   least-update-complexity encoding scheme for MBR codes. A decoding
scheme is  proposed that utilizes the $[n,\alpha]$ RS code to perform data
reconstruction for MSR codes. The proposed decoding scheme has better error correction
capability and incurs the least number of node accesses when errors are present. A new decoding scheme is also proposed for MBR codes that can correct more error-patterns.

\end{abstract}


\section{Introduction}
\label{SEC:Intro}

Cloud storage is gaining popularity as an alternative to enterprise storage
where data is stored in virtualized pools of storage typically hosted by
third-party data centers.  Reliability is a key challenge in the design of
distributed storage systems that provide cloud storage. Both crash-stop and
Byzantine failures (as a result of software bugs and malicious attacks) are
likely to be present during data retrieval. A
crash-stop failure makes a storage node unresponsive to access requests. In
contrast, a Byzantine failure responds to access requests with erroneous data.
To achieve better reliability, one common approach is to replicate
data files on multiple storage nodes in a network. There are two kinds of approaches: duplication (Google)~\cite{GHE03} and  erasure coding~\cite{KUB00,BHA04}. Duplication makes an exact copy of each data and needs lots of storage space. The advantage of this approach is that only one storage node needs to be accessed to obtain the original data. In contrast, in the second approach, erasure coding is employed to encode the original
data and then the encoded data  is distributed to storage nodes. Typically, multiple
storage nodes need to be accessed to recover the original data. One popular class
of erasure codes is the maximum-distance-separable (MDS) codes.
With $[n,k]$ MDS codes such as Reed-Solomon (RS) codes, $k$ data items are encoded and
then distributed to and stored at $n$ storage nodes. A user or a data collector
can retrieve the original data  by accessing {\it any} $k$ of the storage
nodes, a process referred to as {\it data reconstruction}.

%

Any storage node can fail due to hardware or software damage.  Data stored at
the failed nodes need to be recovered (regenerated)  to remain functional to perform data reconstruction. The
process to recover the stored (encoded) data at a storage node is called {\it
data regeneration}.
A
simple way for data regeneration is to first reconstruct the original data
and then recover the data stored at the failed node.  However, it is
not efficient to retrieve the entire $B$ symbols of the original file to
recover a much smaller fraction of data stored at the failed node.
{\it Regenerating codes}, first introduced  in the pioneer works by Dimakis {\it
et al.} in ~\cite{DIM07,DIM10}, allow efficient data regeneration.  To
facilitate data regeneration, each storage node stores $\alpha$ symbols and a
total of $d$ surviving nodes are accessed to retrieve $\beta \le \alpha$
symbols from each node.  A trade-off exists between the storage overhead
and the regeneration (repair) bandwidth needed for data regeneration.  Minimum
Storage Regenerating (MSR) codes first minimize the amount of data stored per
node, and then the repair bandwidth, while Minimum Bandwidth Regenerating (MBR)
codes carry out the minimization in the reverse order. There have been many
works that focus on the design of regenerating
codes~\cite{WU07,WU10,CUL09,WU09,RAS09,PAW11,OGG11,RAS11}. There are two categories of approaches to regenerate data at a failed node. If
the replacement data is exactly the same as that previously stored at the
failed node, we call it  {\it exact regeneration}. Otherwise, if the
replacement data only guarantees the correctness of data reconstruction and
regeneration properties, it is called {\it functional regeneration}. In
practice, exact regeneration is more desirable since there is no need to inform
each node in the network regarding the replacement. Furthermore, it is easy to
keep the codes systematic via exact regeneration, where partial data can be
retrieved without accessing all $k$ nodes. It has been proved that no linear code performing exact regeneration can achieve the MSR point for any ￼$[n,k,d<2k-3]$ ￼ ￼when $\beta$ is normalized to 1~\cite{SHA12}. ￼However, when $B$ approaches infinity, this is achievable for any $k\le d\le n-1$ ~\cite{CAD10}. In this work, we only consider exact regeneration.

There are several existing code constructions of regenerating codes for exact regeneration\cite{WU09,CAD10,SUH11,RAS11}. In~\cite{WU09}, Wu and Dimakis  apply ideas from interference alignment\cite{CAD08,MAD08} to construct the codes for $n = 4$ and $k = 2$. The idea was extended to the more general case of $k<\max\{3,n/2\}$ in~\cite{SUH11}.
In~\cite{RAS11}, Rashmi
{\it et al.} used product-matrix construction to design optimal $[n,k,d\ge 2k-2]$ MSR codes and $[n,k,d]$ MBR codes for exact regeneration. These constructions of exact-regenerating codes are the first for which the code length $n$ can be chosen independently of other parameters. However,  only  crash-stop
failures of storage nodes are considered in~\cite{RAS11}.

The problem of the security of regenerating codes was  considered
in~\cite{PAW11} and  in~\cite{OGG11,HAN12-INFOCOM,RAS12}. In~\cite{PAW11}, the security
problem against eavesdropping and adversarial attack during the data reconstruction and regeneration
processes was considered. Upper bounds on the maximum amount of information that
can be stored safely were derived. Pawar {\it et al.} 
 also gave an explicit code construction  for $d=n-1$ in
the bandwidth-limited regime. The problem of Byzantine fault tolerance for
regenerating codes was considered in~\cite{OGG11}. Oggier and Datta investigated the
resilience of regenerating codes when supporting multi-repairs. By
collaboration among newcomers, they derived upper bounds on the resilience
capability of regenerating codes. Our work  deals with Byzantine failures for product-matrix regenerating codes and it does not need to have multiple newcomers to recover the failures.

Based on the same code construction as given in~\cite{RAS11}, Han {\it et al.} extended Rashmi's work to provide decoding algorithms that can handle
Byzantine failures~\cite{HAN12-INFOCOM}. In~\cite{HAN12-INFOCOM}, decoding algorithms for both MSR and MBR error-correcting product-matrix codes were
provided. In particular,  the decoding of an $[n,k,d]$ MBR code given in~\cite{HAN12-INFOCOM} can decode errors up to error correction capability of $\lfloor\frac{n-k+1}{2}\rfloor=\frac{n-k}{2}$ since $n-k$ is even.  In~\cite{RAS12}, the code capability and resilience were discussed
for error-correcting regenerating codes. Rashmi, {\it et al.} proved that it is possible to decode  an $[n,k,d]$ MBR code up to $\lfloor\frac{n-k}{2}\rfloor$ errors. The authors also claimed  that any $[n,k,d\ge 2k-2]$ MSR code can be decoded up to $\lfloor\frac{n-k}{2}\rfloor$ errors.  However no   explicit decoding  (data reconstruction) procedure was provided due to which these codes cannot be used in practice. Thus, one contribution of this paper is to present a decoding algorithm for MSR codes. 

In addition to bandwidth efficiency and error correction capability, another
desirable feature for regenerating codes is {\it update
complexity}~\cite{RAW11}, defined as the number of nonzero
elements in the row of the encoding matrix with the maximum Hamming weight.\footnote{The update complexity adopted from~\cite{RAW11} is not equivalent to the maximum number of encoded symbols that
must be updated while a single data symbol is modified.}  The smaller the number, the lower the update complexity is.  Low update complexity
is desirable in scenarios where updates are frequent.

One drawback of the decoding algorithms for MSR codes given
in~\cite{HAN12-INFOCOM} is that, when one or more storage nodes have erroneous
data, the decoder needs to access extra data from many storage nodes (at least
$k$ more nodes) for data reconstruction. Furthermore, when one symbol in the
original data is updated, all storage nodes need to update their respective
data. Thus, the MSR and MBR codes in ~\cite{HAN12-INFOCOM} have the maximum
possible update complexity. Both of these issues deficiencies are addressed in this paper.
First, we propose a general encoding scheme for MSR codes. As a special case,
least-update-complexity codes are designed. We also design
least-update-complexity encoding matrix for the MBR  codes by using the
coefficients of generator polynomials of the $[n,k]$ and $[n,d]$ RS codes. The proposed codes are not only with least update complexity but also with the smallest numbers of updated symbols while a single data symbol is modified. This is in contrast to the  existing product-matrix codes.
Second, a new decoding algorithm is presented for MSR codes.  It not only
exhibits better error correction capability but also incurs low communication
overhead when errors occur in the accessed data.  Third, we devise a  decoding scheme for the MBR codes that can correct more
error patterns compared to the one in~\cite{HAN12-INFOCOM}.

The main contributions of this paper beyond the existing literature are as follows:
\begin{itemize}
\item The general encoding schemes of product-matrix MSR
and MBR codes are derived. The encoder based on RS codes is no longer limited to the
Vandermonde matrix proposed in~\cite{RAS11} and \cite{HAN12-INFOCOM}. Any generator matrix of the
corresponding RS codes can be employed for the MSR and MBR codes. As a result,
this highlights the connection between product-matrix MSR and MBR codes and
well-known RS codes in coding theory.
\item The MSR and MBR codes with systematic generator matrices of the RS codes  are provided. These codes have least update complexity compared to existing codes such as systematic MSR and MBR codes proposed by Rashmi {\it et al.}~\cite{RAS11}. This
approach also makes product-matrix MSR and MBR codes more practical due to higher
update efficiency. 
\item The detailed decoding algorithm of data construction of
MSR codes is provided. It is non-trivial to extend the  decoding procedure given
in~\cite{RAS11}  to handle errors. The difficulty arises 
from the fact that an error in $Y_{\alpha\times n}$
will propagate into many places in $P$ and $Q$. Due to the operations involved in
the decoding process, many rows  cannot be decoded
successfully or correctly. No decoding algorithm was provided in~\cite{RAS12} that can
decode up to  $\lceil{(n-k+1)/2}\rceil$ errors even though the error-correction
capability was analyzed in~\cite{RAS12}. 
\item The decoding algorithm of MBR codes that can decode beyond error-correction capability for some error patterns is also presented. This decoding algorithm can correct errors up to $$\frac{n-k}{2}+ \left\lfloor\frac{n-k+1-\lfloor\frac{n-k+1}{2}\rfloor}{2}\right\rfloor$$
even though not all error patterns up to such number of errors can be corrected.
\end{itemize}

The rest of this paper is organized as follows. Section~\ref{SEC:review} gives
an overview of error-correcting regenerating codes.  Section~\ref{SEC:MSR}
presents the least-update-complexity encoding and decoding schemes
for error-correcting MSR regenerating codes.  Section~\ref{SEC:MBR-coding}
demonstrates the least-update-complexity encoding of MBR codes and the
corresponding decoding scheme.  Section~\ref{SEC:eval} details evaluation
results for the proposed decoding schemes. Section~\ref{SEC:conclude} concludes
the paper with a list of future work. Since only error-correcting regenerating
codes are considered in this work, unless stated otherwise, we refer to
error-correcting MSR and MBR codes as MSR and MBR codes in the rest of the
paper.

\section{Error-Correcting  Product-Matrix Regenerating Codes}
\label{SEC:review}
In this section, we give a brief overview of  regenerating codes,  and the
MSR and MBR product-matrix code constructions in \cite{RAS11}.
\subsection{Regenerating Codes}
\label{subSEC:RC}

Let $\alpha$ be the number of symbols stored at each storage node and
$\beta\le\alpha$ the number of symbols downloaded from each storage during
regeneration.  To repair the stored data at the failed node, a helper node
accesses $d$ surviving nodes. The design of regenerating codes ensures that the total  regenerating bandwidth be
much less than that of the original data, $B$. A regenerating code must be capable
of reconstructing the original data symbols and regenerating coded data  at a
failed node.   An $[n,k,d]$ regenerating code requires at least $k$ nodes to ensure successful data reconstruction, and $d$
surviving nodes to perform
regeneration~\cite{RAS11}, where $n$ is the number of storage
nodes and $k\le d\le n-1$.
%

The cut-set bound given in~\cite{WU07,DIM10} provides a constraint on
the  repair bandwidth. By this bound, any regenerating code must satisfy
the following inequality:
\begin{eqnarray}
B\le \sum_{i=0}^{k-1} \min\{\alpha,(d-i)\beta\}~.\label{main-inequality}
\end{eqnarray}
From~\eqref{main-inequality}, $\alpha$ or $\beta$ can be minimized achieving
either the minimum storage requirement or  the minimum repair bandwidth
requirement, but not both. The two extreme points in~\eqref{main-inequality}
are referred to as the minimum storage regeneration (MSR) and minimum bandwidth
regeneration (MBR) points, respectively.  The values of $\alpha$ and $\beta$
for the MSR point can be obtained by first minimizing $\alpha$ and then minimizing
$\beta$:
\begin{eqnarray}
\alpha&=&d-k+1\nn\\
B&=&k(d-k+1)=k\alpha~,\label{NMSR}
\end{eqnarray}
where we normalize $\beta$ and set it equal to $1$.\footnote{It has been proved that when designing  $[n,k,d]$ MSR codes for $k/(n+1)\le 1/2$. it
suffices to consider those with $\beta=1$~\cite{RAS11}.} Reversing the order of minimization we have   $\alpha$ for MBR as
\begin{eqnarray}
\alpha&=&d\nn\\
B&=&kd-k(k-1)/2~,\label{NMBR}
\end{eqnarray}
while $\beta=1$.

\subsection{Product-Matrix MSR  Codes With Error Correction Capability}
\label{subSEC:MSR}
Next, we describe the MSR code construction originally given in~\cite{RAS11} and adapted later in~\cite{HAN12-INFOCOM}. Here, we assume $d = 2\alpha$.\footnote{An elegant  method to extend the construction of $d>2\alpha$ based on the construction of $d=2\alpha$ has been given in~\cite{RAS11}. Since the same technology can be applied to the code constructions proposed in this work, it is omitted here.}
The information sequence $\m=[m_0,m_1,\ldots,
m_{B-1}]$ can be arranged into
an information vector $U=\left[Z_1Z_2\right]$ with size $\alpha\times d$ such that
 $Z_1$ and $Z_2$ are symmetric matrices with
dimension $\alpha\times\alpha$.
 An $[n,d=2\alpha]$ RS code is adopted to construct the
MSR code~\cite{RAS11}. Let $a$ be a generator of $GF(2^m)$. In the encoding of the MSR code, we have
\begin{eqnarray}
U\cdot G=C,\label{eq:generator}
\end{eqnarray}
where
 $$G=\left[\begin{array}{cccc}
1&1&\cdots&1\\
a^0&a^1&\cdots&a^{n-1}\\(a^0)^2&(a^1)^2&\cdots&(a^{n-1})^2\\
&&\vdots&\\
(a^0)^{d-1}&(a^1)^{d-1}&\cdots&(a^{n-1})^{d-1}\end{array}\right],$$
and $C$  is
the codeword vector with dimension $(\alpha\times n)$.


It is possible to rewrite generator matrix $G$ of the RS code as,
 \begin{eqnarray}
G&=&
\left[\begin{array}{cccc}
1&1&\cdots&1\\
a^0&a^1&\cdots&a^{n-1}\\(a^0)^2&(a^1)^2&\cdots&(a^{n-1})^2\\
&&\vdots&\\
(a^0)^{\alpha-1}&(a^1)^{\alpha-1}&\cdots&(a^{n-1})^{\alpha-1}\\
(a^0)^\alpha 1&(a^1)^\alpha 1&\cdots&(a^{n-1})^\alpha 1\\
(a^0)^\alpha a^0&(a^1)^\alpha a^1&\cdots&(a^{n-1})^\alpha a^{n-1}\\
(a^0)^\alpha(a^0)^2&(a^1)^\alpha(a^1)^2&\cdots&(a^{n-1})^\alpha(a^{n-1})^2\\
&&\vdots&\\
(a^0)^\alpha (a^0)^{\alpha-1}&(a^1)^\alpha(a^1)^{\alpha-1}&\cdots&(a^{n-1})^\alpha(a^{n-1})^{\alpha-1}
\end{array}\right]\\
&=&\left[\begin{array}{c}
\bar{G}\\
\bar{G}\Delta
\end{array}
\right]~,
\label{MSR-encoding}
\end{eqnarray}
where $\bar{G}$ contains the first $\alpha$ rows in $G$, and $\Delta$ is a
diagonal matrix with $(a^0)^\alpha,\ (a^1)^\alpha,\ (a^2)^\alpha,\ldots,\
(a^{n-1})^\alpha$ as diagonal elements, namely,
\begin{equation}
\label{eq:delta}
\Delta=
\left[\begin{array}{cccccc}
(a^0)^\alpha&0&0&\cdots&0&0\\
0&(a^1)^\alpha&0&\cdots&0&0\\
&&\vdots&&\\
0&0&0&\cdots&0&(a^{n-1})^\alpha
\end{array}\right]~.
\end{equation}
Note that if the RS  code is over $GF(2^m)$ for $m\ge \lceil
\log_2 n\alpha\rceil$, then it can be shown that $(a^0)^\alpha,\ (a^1)^\alpha,\
(a^2)^\alpha,\ldots,\ (a^{n-1})^\alpha$ are all distinct.
According to the encoding procedure, the $\alpha$ symbols stored in storage node $i$ are given by,
$$U\cdot \left[\begin{array}{c}
\gggg_i^T\\
(a^{i-1})^\alpha \gggg_i^T\end{array}\right]=Z_1\gggg_i^T+(a^{i-1})^\alpha Z_2\gggg_i^T,$$
where $\gggg_i^T$ is the $i$th column in $\bar{G}$.

\subsection{Product-Matrix MBR  Codes With Error Correction Capability}
\label{sec:MBR-encoding}
In this section, we describe the MBR code constructed in~\cite {RAS11} and reformatted later in~\cite{HAN12-INFOCOM}. Note that at the MBR point, $\alpha=d$.
Let  the information sequence $\m=[m_0,m_1,\ldots, m_{B-1}]$ be arranged into
an information vector $U$ with size $\alpha\times d$, where
\begin{eqnarray}
\label{U-mbr}
U=\left[\begin{array}{cc}
A_1&A_2^T\\
A_2&\0\end{array}
\right]~,
\end{eqnarray}
$A_1$ is a $k\times k$ symmetric matrix, $A_2$ a $(d-k)\times  k$ matrix, $\0$ is
the $(d-k)\times (d-k)$ zero matrix.  Note that both $A_1$ and $U$ are symmetric.  It
is clear that $U$ has a dimension $d\times d$ (or $\alpha\times d$).
An $[n,d]$ RS code is chosen to encode each row of $U$. The generator matrix of the RS code is given as
\begin{eqnarray}
\label{MBR-G}
G&=&\left[\begin{array}{cccc}
1&1&\cdots&1\\
a^0&a^1&\cdots&a^{n-1}\\
(a^0)^2&(a^1)^2&\cdots&(a^{n-1})^2\\
&&\vdots&\\
(a^0)^{k-1}&(a^1)^{k-1}&\cdots&(a^{n-1})^{k-1}\\
(a^0)^k&(a^1)^k&\cdots&(a^{n-1})^k\\
&&\vdots&\\
(a^0)^{d-1}&(a^1)^{d-1}&\cdots&(a^{n-1})^{d-1}\end{array}\right]~,
\end{eqnarray}
 where $a$ is a generator of $GF(2^m)$.  Let $C$ be the codeword vector with
dimension $(\alpha\times n)$. It can be obtained as $$U\cdot G=C.$$ From~\eqref{MBR-G}, $G$ can be divided into two sub-matrices as
\begin{eqnarray}
G=\left[\begin{array}{c}
G_k\\
S
\end{array}\right]~,\label{MBR-G-2}
\end{eqnarray}
where
\begin{eqnarray}
G_k
=\left[\begin{array}{cccc}
1&1&\cdots&1\\
a^0&a^1&\cdots&a^{n-1}\\
(a^0)^2&(a^1)^2&\cdots&(a^{n-1})^2\\
&&\vdots&\\
(a^0)^{k-1}&(a^1)^{k-1}&\cdots&(a^{n-1})^{k-1}\\
\end{array}\right]\label{G-k}
\end{eqnarray}
and
$$S=\left[\begin{array}{cccc}(a^0)^{k}&(a^1)^{k}&\cdots&(a^{n-1})^{k}\\
&&\vdots&\\
(a^0)^{d-1}&(a^1)^{d-1}&\cdots&(a^{n-1})^{d-1}\end{array}\right]~.$$ It can be shown that
$G_k$ is a generator matrix of the $[n,k]$ RS code and it will be used in the
decoding  for data reconstruction.

\section{Encoding and Decoding Schemes for Product-Matrix MSR Codes}
\label{SEC:MSR}
In this section, we propose a new encoding scheme for $[n,d]$
error-correcting MSR codes. With a feasible matrix $\Delta$, $\bar G$
in~\eqref{MSR-encoding} can be any generator matrix of the $[n,\alpha]$ RS
code. The code construction in~\cite{RAS11,HAN12-INFOCOM} is thus a special case of
our proposed scheme. We can also select a suitable generator matrix such that
the update complexity of the resulting code is minimized. A decoding
scheme is then proposed that  uses the subcode of the $[n,d]$ RS code, the
$[n,\alpha=k-1]$ RS code generated by $\bar{G}$, to perform the data
reconstruction.

\subsection{Encoding Schemes for Error-Correcting MSR Codes}
\label{SEC:MSR-encoding}
RS codes are known to have very fast decoding algorithms and exhibit good error correction capability.
From \eqref{MSR-encoding} in Section~\ref{subSEC:MSR}, a generator matrix $G$ for product-matrix MSR codes needs to satisfy:
\begin{enumerate}
\item $G=\left[\begin{array}{c}
\bar G\\
\bar G\Delta\end{array}\right],$ where $\bar G$ contains the first $\alpha$ rows in $G$ and $\Delta$ is a diagonal matrix with distinct elements in the diagonal.
\item $\bar G$ is a generator matrix of the $[n,\alpha]$ RS code and $G$ is a generator matrix of the $[n,d=2\alpha]$ RS code.
\end{enumerate}
Next, we present a sufficient condition for $\bar G$ and $\Delta$ such that $G$
is a generator matrix of an $[n,d]$ RS code. We first introduce some notations.
Let $g_{0y}(x)=\prod_{i=0}^{n-y-1}(x-a^i)$ and the $[n,y]$ RS code generated by
$g_{0y}(x)$ be $C_{0y}$. Similarly, let $g_{1y}(x)=\prod_{i=1}^{n-y}(x-a^i)$
and the $[n,y]$ RS code generated by $g_{1y}(x)$ be $C_{1y}$. Clearly,
$a^0,a^1,a^2,\ldots,a^{n-y-1}$ are roots of $g_{0y}(x)$, and
$a^1,a^2,\ldots,a^{n-y}$ are roots of $g_{1y}(x)$. Thus, $C_{0y}$ and $C_{1y}$
are equivalent RS codes.

\begin{theorem}
\label{thm:MSR-encoding}
Let $\bar G$ be a generator matrix of  the $[n, \alpha]$ RS code $C_{0\alpha}$.
Let the diagonal elements of $\Delta$ be $b_0,b_1,\ldots,b_{n-1}$ such that
$b_i\neq b_j$ for all $i\neq j$, and  $( b_0,b_1,\ldots,b_{n-1})$ is a codeword in
$C_{1(\alpha+1)}$ but not $C_{1\alpha}$. In other words, $( b_0,b_1,\ldots,b_{n-1})\in
C_{1(\alpha+1)}\backslash C_{1\alpha}$.   Then,  $G = \left[\begin{array}{c}
\bar G\\
\bar G\Delta\end{array}\right]$ is a generator matrix of the $[n,d]$ RS code
$C_{0d}$.
\end{theorem}

\begin{proof}
We need to prove that each row of $\bar G\Delta$ is a codeword of $C_{0d}$ and
all rows in $G$ are linearly independent. Let $\hat C_{0\alpha}$ be the  dual
code of $C_{0\alpha}$. It is well-known that $\hat C_{0\alpha}$ is an
$[n,n-\alpha]$ RS code~\cite{LIN04,MOO05}. Similarly, let $\hat C_{0d}$ be the
dual code of  $C_{0d}$ and its generator matrix be $H_d$. Note that $H_d$ is a
parity-check matrix of $C_{0d}$. Let $h_d(x)=(x^n-1)/g_{0d}(x)$ and
$h_\alpha(x)=(x^n-1)/g_{0\alpha}(x)$. Then, the roots of $h_d(x)$ and
$h_\alpha(x)$ are $a^{n-d},a^{n-d+1},\ldots,a^{n-1}$ and
$a^{n-\alpha},a^{n-\alpha+1},\ldots,a^{n-1}$, respectively. Since an RS code is
also a cyclic code,  the generator polynomials of $\hat C_{0d}$ and $\hat
C_{0\alpha}$ are $\hat h_d(x)$ and $\hat h_\alpha(x)$, respectively, where
$\hat h_d(x)=x^{n-d}h_d(x^{-1})$ and $\hat
h_\alpha(x)=x^{n-\alpha}h_\alpha(x^{-1})$. Clearly, the roots of $\hat
h_d(x)$ are $a^{-(n-d)},a^{-(n-d+1)},\ldots,a^{-(n-1)}$  that are equivalent to
$a^{d},a^{d-1},\ldots, a^1$. Similarly, the roots of $\hat h_\alpha(x)$ are
$a^{\alpha},a^{\alpha-1},\ldots, a^1$. Since $\hat h_d(x)$ has roots of
$a^{d},a^{d-1},\ldots, a^1$, we can choose
\begin{eqnarray}
H_d&=&\left[\begin{array}{cccc}
1&1&\cdots&1\\
a^0&a^1&\cdots&a^{n-1}\\
(a^0)^2&(a^1)^2&\cdots&(a^{n-1})^2\\
&&\vdots&\\
(a^0)^{n-d-1}&(a^1)^{n-d-1}&\cdots&(a^{n-1})^{n-d-1}\end{array}\right]
\end{eqnarray}
as the generator matrix of $\hat C_{0d}$. To prove that each row of $\bar
G\Delta$ is a codeword of the RS code $C_{0d}$ generated by $G$, it is sufficient to show
that $\bar G\Delta H_d^T=\0$. From the symmetry of $\Delta$, we have
$$\bar G\Delta H_d^T=\bar G\left(H_d\Delta\right)^T.$$ Thus, we only need to prove that  each row of $H_d\Delta$ is a
codeword in $\hat C_{0\alpha}$. Let the diagonal elements of $\Delta$ be $b_0,
b_1,\ldots,b_{n-1}$. The $i$th row of $H_d\Delta$ is thus
$r_i(x)=\sum_{j=0}^{n-1}b_j(a^j)^{i-1}x^j$ in the polynomial representation.
Let  $(b_0,
b_1,\ldots,b_{n-1})$ be a codeword in
$C_{1(\alpha+1)}$. Then, we have
\begin{eqnarray}
\label{eq:delta}
\sum_{j=0}^{n-1}b_j(a^{\ell'})^j=0\mbox{ for }1\le \ell'\le n-\alpha-1~.
\end{eqnarray}
Substituting $x=a^\ell$, for $1\le \ell\le \alpha$, into $r_i(x)$, it becomes
\begin{eqnarray}
r_i(a^\ell)=\sum_{j=0}^{n-1}b_j(a^j)^{i-1}(a^\ell)^j=\sum_{j=0}^{n-1}b_j(a^{i-1+\ell})^j~.\label{eq:b}
\end{eqnarray}
Let $\ell'=i-1+\ell$. Since $1\le i\le n-d$ and $1\le \ell\le \alpha$, $1\le \ell'\le n-\alpha-1$. By ~\eqref{eq:delta}, $r_i(a^\ell)=0$ for $1\le i\le n-d$ and $1\le \ell\le \alpha$. Hence,
 each row of $H_d\Delta$ is a codeword in $\hat C_{0\alpha}$.

The  $b_i$s need to make all rows in $G$ linearly independent. Since all rows
in $\bar G$ or those in $\bar G\Delta$ are linearly independent, it is sufficient
to prove that $C_{0\alpha}\cap C_{\Delta}=\{\0\}$, where $C_\Delta$ is the code
generated by $\bar G\Delta$. Let $\ccc'$ be a codeword in $C_\Delta$.
$\ccc'=\ccc\Delta$ for some $\ccc\in C_{0\alpha}$.
 It can be shown that,
by the Mattson-Solomon polynomial~\cite{MAC77},
 we can choose
\begin{eqnarray}
\label{MSR-G-bar}
\bar G&=&\left[\begin{array}{cccc}
(a^0)^1&(a^1)^1&\cdots&(a^{n-1})^1\\
(a^0)^2&(a^1)^2&\cdots&(a^{n-1})^2\\
&&\vdots&\\
(a^0)^{\alpha}&(a^1)^{\alpha}&\cdots&(a^{n-1})^{\alpha}\end{array}\right]
\end{eqnarray}
as the generator matrix of $C_{0\alpha}$. Then
$$\ccc'=\uu\bar G\Delta$$ for some $\uu=[u_0,u_1,\ldots,u_{\alpha}]$. Evaluating $\ccc'(x)$ at $a^{0},a^{1},\ldots,a^{n-\alpha-1}$ and putting them into a matrix form, we have
\begin{equation}
\uu\bar G\Delta\tilde G=\zz~,\label{uGD}
\end{equation}
where
$$\tilde G=\left[\begin{array}{cccc}
(a^0)^0&(a^1)^0&\cdots&(a^{n-\alpha-1})^0\\
(a^0)^1&(a^1)^1&\cdots&(a^{n-\alpha-1})^1\\
&&\vdots&\\
(a^0)^{n-1}&(a^1)^{n-1}&\cdots&(a^{n-\alpha-1})^{n-1}\end{array}\right]$$ and $\zz$ is an $(n-\alpha)$-dimensional vector. If $\zz=\0$, then $\ccc\Delta\in C_{0\alpha}$; otherwise, $\ccc\Delta\not\in C_{0\alpha}$. Taking transpose on both sizes of \eqref{uGD}, it becomes
\begin{eqnarray}
&&\tilde G^T\Delta\bar G^T\uu^T\nonumber\\
&=&\left[\begin{array}{cccc}
\sum_{j=0}^{n-1}b_ja^j&\sum_{j=0}^{n-1}b_j(a^2)^j&\cdots&\sum_{j=0}^{n-1}b_j(a^\alpha)^j\\
\sum_{j=0}^{n-1}b_j(a^2)^j&\sum_{j=0}^{n-1}b_j(a^3)^j&\cdots&\sum_{j=0}^{n-1}b_j(a^{\alpha+1})^j\\
&&\vdots&\\
\sum_{j=0}^{n-1}b_j(a^{n-\alpha})^j&\sum_{j=0}^{n-1}b_j(a^{n-\alpha+1})^j&\cdots&\sum_{j=0}^{n-1}b_j(a^{n-1})^j\end{array}\right]\left[\begin{array}{c}
u_0\\
u_1\\
\vdots\\
u_{\alpha-1}\end{array}\right]
=\zz^T~.\label{GDGuT}
\end{eqnarray}
Since $(b_0,b_1,\ldots,b_{n-1})\in C_{1(\alpha+1)}$,
\begin{eqnarray}
\label{eq:alpha}
\sum_{j=0}^{n-1}b_j(a^{\ell})^j=0\mbox{ for }1\le \ell\le n-\alpha-1~.
\end{eqnarray}
Substituting \eqref{eq:alpha} into \eqref{GDGuT} and taking out rows with all zeros, we have
{\footnotesize \begin{eqnarray}
&&\left[\begin{array}{ccccc}
0&0&\cdots&0&\sum_{j=0}^{n-1}b_j(a^{n-\alpha})^j\\
0&0&\cdots&\sum_{j=0}^{n-1}b_j(a^{n-\alpha})^j&\sum_{j=0}^{n-1}b_j(a^{n-\alpha+1})^j\\\
&&\vdots&\\
\sum_{j=0}^{n-1}b_j(a^{n-\alpha})^j&\sum_{j=0}^{n-1}b_j(a^{n-\alpha+1})^j&\cdots&\sum_{j=0}^{n-2}b_j(a^{n-2})^j&\sum_{j=0}^{n-1}b_j(a^{n-1})^j\end{array}\right]\left[\begin{array}{c}
u_0\\
u_1\\
\vdots\\
u_{\alpha-1}\end{array}\right]\nonumber\\
&=&\left[\begin{array}{c}
z_{n-2\alpha}\\
z_{n-2\alpha+1}\\
\vdots\\
z_{n\alpha-1}\end{array}\right]=\tilde \zz~.\label{GDGuT-2}
\end{eqnarray}}

If
$\sum_{j=0}^{n-1}b_j(a^{n-\alpha})^j=0$, i.e., $a^{n-\alpha}$ is a root of
$\sum_{j=0}^{n-1}b_jx^j$, then $\ccc'=[1,0,\ldots,0]\bar G\Delta\in C_{0\alpha}$ due to the fact that $\uu=[1,0,\ldots,0]$ makes $\tilde \zz=\0$ in~\eqref{GDGuT-2}. Thus, we need to exclude the codewords in $C_{1(\alpha+1)}$
that have $a^{n-\alpha}$ as a root. These codewords turn out to be in $C_{1\alpha}$. If $\sum_{j=0}^{n-1}b_j(a^{n-\alpha})^j\neq 0$, then it is clear that the only $\uu$ making $\tilde \zz=\0$ in~\eqref{GDGuT-2} is the all-zero vector. Hence, any $( b_0,b_1,\ldots,b_{n-1})\in
C_{1(\alpha+1)}\backslash C_{1\alpha}$ does not make $\tilde \zz$ zero except $\uu=\0$.
%
%
\end{proof}

\begin{corollary}
\label{col:G}
Under the condition that the RS  code is over $GF(2^m)$ for $m\ge \lceil \log_2
n\rceil$ and $\gcd(2^ m-1,\alpha)=1$, the diagonal elements of $\Delta$, $b_0,
b_1,\ldots, b_{n-1}$, can be
$$\gamma(a^0)^\alpha,\gamma(a)^\alpha,\gamma(a^2)^\alpha,\ldots,\gamma(a^{n-1})^\alpha~,$$
where $\gamma\in GF(2^m)\backslash\{\0\}$.
\end{corollary}
\begin{proof}
Note that one valid generator matrix of  $C_{1(\alpha+1)}$  is
\begin{eqnarray}
\label{MSR-G-alpha}
\left[\begin{array}{cccc}
1&1&\cdots&1\\
a^0&a^1&\cdots&a^{n-1}\\
(a^0)^2&(a^1)^2&\cdots&(a^{n-1})^2\\
&&\vdots&\\
(a^0)^{\alpha}&(a^1)^{\alpha}&\cdots&(a^{n-1})^{\alpha}\end{array}\right].
\end{eqnarray}
$(b_0,b_1,\ldots,b_{n-1})\in C_{1(\alpha+1)}\backslash C_{1\alpha}$ can be
represented as
$b_i=\gamma(a^i)^\alpha+f_i$, where $(f_0,f_1,\ldots,f_{n-1})\in C_{1,\alpha}$. Now choose $(f_0,f_1,\ldots,f_{n-1})$ to be all-zero codeword. Under the condition that the RS  code is over $GF(2^m)$ for $m\ge \lceil \log_2 n\rceil$ and $\gcd(2^m-1,\alpha)=1$, $\gamma(a^0)^\alpha,\gamma(a)^\alpha,\gamma(a^2)^\alpha,\ldots,\gamma(a^{n-1})^\alpha$ is equivalent to $\gamma(a^\alpha)^0,\gamma(a^\alpha)^1,\gamma(a^\alpha)^2,\ldots,\gamma(a^\alpha)^{n-1}$. If $a^\alpha$ is a generator of $GF(2^m)$, then all elements of  $\gamma(a^\alpha)^0,\gamma(a^\alpha)^1,\gamma(a^\alpha)^2,\ldots,\gamma(a^\alpha)^{n-1}$ are distinct.  It is well-known that $a^\alpha$ is a generator if $\gcd(2^m-1,\alpha)=1$.
\end{proof}

It is clear that by setting $\gamma=1$ in Corollary~\ref{col:G}, we obtain the
generator matrix $G$ given in~\eqref{MSR-encoding} first proposed
in~\cite{RAS11,HAN12-INFOCOM} as a special case.\footnote{Even though the roots in
$G$ given in ~\eqref{MSR-encoding} are different from those for the proposed
generator matrix, they generate equivalent RS codes.}

One advantage of the proposed scheme is that it can now operate on a smaller
finite field than that of the scheme in~\cite{RAS11,HAN12-INFOCOM}.  Another
advantage is that one can choose $\bar{G}$ (and $\Delta$ accordingly) freely as
long as $\bar{G}$ is the generator matrix of an $[n, \alpha]$ RS code. In
particular, as discussed in Section~\ref{SEC:Intro}, to minimize the update
complexity, it is desirable to choose a generator matrix that has the least
row-wise maximum Hamming weight. Next, we present a least-update-complexity
generator matrix that satisfies~\eqref{MSR-encoding}.

\begin{corollary}
Suppose $\Delta$ is chosen according to Corollary~\ref{col:G}. Let $\bar G$ be
the generator matrix associated with  a systematic  $[n,\alpha]$ RS code. That is,
\begin{eqnarray}
 {\bar G}
&=&\left[\begin{array}{cccccccccc}
b_{00}&b_{01}&b_{02}&\cdots&b_{0(n-\alpha-1)}&1&0&0&\cdots&0\\
b_{10}&b_{11}&b_{12}&\cdots&b_{1(n-\alpha-1)}&0&1&0&\cdots&0\\
b_{20}&b_{21}&b_{22}&\cdots&b_{2(n-\alpha-1)}&0&&1&\cdots&0\\
&\vdots&&&&\vdots&&&&\vdots\\
b_{(\alpha-1)0}&b_{(\alpha-1)1}&b_{(\alpha-1)2}&\cdots&b_{(\alpha-1)(n-\alpha-1)}&0&0&0&\cdots&1
\end{array}\right]~,\label{MSR-G-S}\end{eqnarray}
where
$$x^{n-\alpha+i}=u_i(x)g(x)+b_i(x)\mbox{ for } 0\le i\le \alpha-1$$
and
$$b_i(x)=b_{i0}+b_{i1}x+\cdots+b_{i(n-\alpha-1)}x^{n-\alpha-1}~.$$
Then,  $G=\left[\begin{array}{c}
\bar G\\
\bar G\Delta\end{array}\right]$ is a least-update-complexity generator matrix.
\end{corollary}
\begin{proof}
The result holds since each row of $\bar G$ is a nonzero codeword with
the minimum Hamming weight $n-\alpha+1$.
\end{proof}

The update complexity adopted from~\cite{RAW11} is not equivalent to the maximum number of encoded symbols that
must be updated when a single data symbol is modified. If  the modified data symbol is located in the diagonal of $Z_1$ or $Z_2$,  $(n-\alpha+1)$ encoded symbols need to be updated; otherwise, there are two corresponding encoding symbols in $U$ modified such that $2(n-\alpha+1)$ encoded symbols need to be updated.
%

\subsection{Decoding Scheme for MSR Codes}
\label{SEC:decoding}
Unlike the decoding scheme in~\cite{HAN12-INFOCOM} that uses $[n,d]$ RS
code, we propose to use the subcode of the $[n,d]$ RS code, i.e., the
$[n,\alpha=k-1]$ RS code generated by $\bar{G}$, to perform data
reconstruction. The advantage of using  the $[n,k-1]$ RS code is two-fold.
First,  its error correction capability is higher. Specifically, it can tolerate
$\lfloor\frac{n-k}{2}\rfloor$ instead of $\lfloor\frac{n-d}{2}\rfloor$
errors. Second, it only requires the access of two additional storage nodes
(as opposed to $d-k+2=k$ nodes) for each extra error.

Without loss of generality, we assume that the data collector retrieves encoded
symbols from $k+2v$ ($v\ge 0$) storage nodes, $j_0,j_1,\ldots,j_{k+2v-1}$.  We
also assume that  there are $v$ storage nodes  whose received symbols are
erroneous.  The stored information on the $k+2v$ storage nodes are collected as
the  $k+2v$ columns in $Y_{\alpha\times (k+2v)}$. The $k+2v$ columns of $G$
corresponding to storage nodes $j_0,j_1,\ldots,j_{k+2v-1}$ are denoted as the
columns of $G_{k+2v}$.  First, we discuss data reconstruction when $v=0$. The
decoding procedure is similar to that in~\cite{RAS11}.

\paragraph*{No Error}
In this case, $v=0$ and there is no error in  $Y$.  Then,
\begin{eqnarray}
Y_{\alpha\times k}&=&UG_{k}\nonumber\\
&=&[Z_1Z_2]\left[\begin{array}{c}\bar{G}_{k}\nonumber\\
\bar{G}_{k}\Delta\end{array}\right]\\&=&[Z_1\bar{G}_{k}+Z_2\bar{G}_{k}\Delta]~.\label{UG-no-error}
\end{eqnarray}
Multiplying $\bar{G}_{k}^T$ to both sides of \eqref{UG-no-error}, we have~\cite{RAS11},
{\small \begin{eqnarray}
\bar{G}_{k}^TY_{\alpha\times k}&=&\bar{G}_{k}^TUG_{k}\nonumber\\
&=&[\bar{G}_{k}^TZ_1\bar{G}_{k}+\bar{G}_{k}^TZ_2\bar{G}_{k}\Delta]\nonumber\\
&=&P+Q\Delta~.\label{PQ-no-error}
\end{eqnarray}}

Since $Z_1$ and $Z_2$  are symmetric, $P$ and $Q$ are symmetric as well. The
$(i,j)$th element of $P+Q\Delta$, $1\le i,j\le k$ and $i\neq j$,  is
\begin{eqnarray}
p_{ij}+q_{ij}a^{(j-1)\alpha}~,\label{pq-ij}
\end{eqnarray}
and the $(j,i)$th element is given by
\begin{eqnarray}
p_{ji}+q_{ji}a^{(i-1)\alpha}~.\label{pq-ji}
\end{eqnarray}
Since $a^{(j-1)\alpha}\neq a^{(i-1)\alpha}$ for all $i\neq j$, $p_{ij}=p_{ji}$,
and $q_{ij}=q_{ji}$, combining \eqref{pq-ij} and \eqref{pq-ji}, the values of
$p_{ij}$ and $q_{ij}$ can be obtained. Note that we only obtain $k-1$ values
for each row of $P$ and $Q$ since no elements in the diagonal of $P$ or $Q$ are
obtained.

To decode $P$, recall that $P=\bar{G}_{k}^TZ_1\bar{G}_{k}$. $P$ can
be treated as a portion of the codeword vector, $\bar{G}_{k}^TZ_1\bar{G}$. By
the construction of $\bar{G}$, it is easy to see that $\bar{G}$ is a generator
matrix of the $[n,k-1]$ RS code. Hence, each row in the matrix
$\bar{G}_{k}^TZ_1\bar{G}$ is a codeword. Since we know $k-1$ components
in each row of $P$, it is possible to decode $\bar{G}_{k}^TZ_1\bar{G}$ by the
error-and-erasure decoder of the $[n,k-1]$ RS code.\footnote{ The
error-and-erasure decoder of an $[n,k-1]$ RS code can successfully decode a
received vector if $s+2v<n-k+2$, where $s$ is the number of erasure (no symbol)
positions, $v$ is the number of errors in the received portion of the received
vector, and $n-k+2$ is the minimum Hamming distance of the $[n,k-1]$ RS code.}

 Since one cannot locate any erroneous position from the decoded rows of $P$, the decoded $\alpha$ codewords are accepted as $\bar{G}_{k}^TZ_1\bar{G}$. By collecting the last $\alpha$ columns of $\bar{G}$ as $\bar{G}_\alpha$ to find its inverse (here it is an identity matrix), one can recover $\bar{G}_{k}^TZ_1$ from $\bar{G}_{k}^TZ_1\bar{G}$.
Since any $\alpha$ rows in $\bar{G}_{k}^T$ are independent and thus invertible, we can pick any $\alpha$ of them to recover $Z_1$.
$Z_2$ can be obtained similarly by $Q$.

It is not trivial to extend the above decoding procedure to the case of errors. The difficulty is raised from the fact that for any error in $Y_{\alpha\times n}$, this error will propagate into many places in $P$ and $Q$, due to operations involved in \eqref{PQ-no-error}, \eqref{pq-ij}, and \eqref{pq-ji}, such that many rows of them cannot be decoded successfully or correctly (Please refer to Lemma~\ref{lemma}). In the following we present how to locate erroneous columns in $Y$ based on RS decoder.

\vspace{0.5cm}

\noindent{\bf Single Error:} In this case, $v=1$ and only one column of  $Y_{\alpha\times(k+2)}$ is
erroneous. Without loss of generality, we assume the erroneous column is the
first column in $Y$. That is, the symbols received from storage node $j_0$
contain error. Let $E=\left[\eee^T_1|\0\right]$ be the error matrix, where
$\eee_1=[e_{11},e_{12,},\ldots, e_{1\alpha}]$ and $\0$ is all-zero matrix with
dimension $\alpha\times (k+1)$. Then
\begin{eqnarray}
Y_{\alpha\times(k+2)}&=&UG_{k+2}+E\nonumber\\
&=&[Z_1Z_2]\left[\begin{array}{c}\bar{G}_{k+2}\nonumber\\
\bar{G}_{k+2}\Delta\end{array}\right]+E\\&=&[Z_1\bar{G}_{k+2}+Z_2\bar{G}_{k+2}\Delta]+E~.\label{UG-error}
\end{eqnarray}
Multiplying $\bar{G}_{k+2}^T$ to both sides of \eqref{UG-error}, we have
 \begin{eqnarray}
\bar{G}_{k+2}^TY_{\alpha\times(k+2)}&=&\bar{G}_{k+2}^TUG_{k+2}+\bar{G}_{k+2}^TE\nonumber\\
&=&[\bar{G}_{k+2}^TZ_1\bar{G}_{k+2}+\bar{G}_{k+2}^TZ_2\bar{G}_{k+2}\Delta]+\bar{G}_{k+2}^TE\nonumber\\
&=&P+Q\Delta+\left[\bar{G}_{k+2}^T\eee^T_1|\0\right]\nonumber\\
&=&\tilde{P}+\tilde{Q}\Delta~.\label{PQ-error}
\end{eqnarray}

It is easy to see that the errors only affect the first column of
$\tilde{P}+\tilde{Q}\Delta$ since the nonzero elements are all in the first
column of $\left[\bar{G}_{k+2}^T\eee^T_1|\0\right]$.  Similar to  \eqref{pq-ij}
and \eqref{pq-ji}, the values of $\tilde{p}_{ij}$ and $\tilde{q}_{ij}$, where
$i\neq j$, are obtained from $\bar{G}_{k+2}^TY_{\alpha\times(k+2)}$ even though there are some
errors in them. Note that we only obtain $k+1$ values for each row of
$\tilde{P}$ and $\tilde{Q}$. Since the $(j,1)$th elements of $\bar{G}_{k+2}^TY_{\alpha\times(k+2)}$
may be erroneous for $1\le j\le k+2$, the values calculated from them contain
errors as well. Then the first column and the first row of $\tilde{P}$
($\tilde{Q}$)  have errors. Note that each row of $\tilde{P}$ ($\tilde{Q}$)
has only at most one error except the first row.

First, we decode $\tilde{P}$. Recall that $P=\bar{G}_{k+2}^TZ_1\bar{G}_{k+2}$.
As mentioned earlier, $P$ can be treated as a portion of the codeword vector
$\bar{G}_{k+2}^TZ_1\bar{G}$, and  then $\tilde{P}$ can be decoded by  the
$[n,k-1]$ RS code.  Since we have obtained $k+1$ components in each row of
$\tilde{P}$, it is possible to correctly decode each row of
$\bar{G}_{k+2}^TZ_1\bar{G}$, except for the first row of $\tilde{P}$, using the
error-and-erasure decoder of the RS code.

Let $\hat{P}$ be the corresponding portion of decoded codeword vector to
$\tilde{P}$ and  $E_P=\hat{P}\oplus \tilde{P}$ be the error pattern vector.
Next we describe how to  locate the incorrect  row after decoding every row (in
this case we assume that the error occurs in the first row).  Now suppose that there are more than two
errors in the first column of $\tilde{P}$.\footnote{It will be shown later that
the number of errors in the first column of $\tilde{P}$ is at least three.} Let
these errors be in $(j_1,1)$th, $(j_2,1)$th,$\cdots$, and $(j_\ell,1)$th
positions in $\tilde{P}$. After decoding all rows of $\tilde{P}$, it is easy to
see that all rows but the first row can be decoded correctly due to at most
one error occurring in each row.
Then one can confirm that the number of nonzero elements in $E_P$ in the first column is at least three since only the error in the first position of the first column can be decoded incorrectly. Other than the first column in $E_P$ there is at most one nonzero element in rest of the columns.
Then the first column in $\hat{P}$ has correct elements except the one in the first row. Just copy all elements in the first column of $\hat{P}$ to those corresponding positions of its first row to make $\hat{P}$  a symmetric matrix. We then collect any $\alpha$ columns of $\hat{P}$  except the first column as $\hat{P}_\alpha$ and find its corresponding $\bar{G}_\alpha$. By multiplying  the inverse of $\bar{G}_\alpha$ to $\hat{P}_\alpha$, one can recover $\bar{G}_{k+2}^TZ_1$.
Since any $\alpha$ rows in $\bar{G}_{k+2}^T$ are independent and thus invertible, we can pick any $\alpha$ of them to recover $Z_1$.
$Z_2$ can be obtained similarly by $Q$.

\vspace{0.5cm}

\noindent{\bf Multiple  Errors:}
Before presenting the proposed decoding algorithm, we first prove that a
decoding procedure can always successfully decode $Z_1$ and $Z_2$ if $v\le
\lfloor\frac{n-k}{2}\rfloor$ and all storage nodes are accessed.
Assume the storage nodes with errors correspond to the $\ell_0$th,
$\ell_1$th, $\ldots$, $\ell_{v-1}$th columns in the received matrix
$Y_{\alpha\times n}$. Then,
 \begin{eqnarray}
&&\bar{G}^TY_{\alpha\times n}\nonumber\\
&=&\bar{G}^TUG+\bar{G}^TE\nonumber\\
&=&\bar{G}^T[Z_1Z_2]\left[\begin{array}{c}\bar{G}\nonumber\\
\bar{G}\Delta\end{array}\right]+\bar{G}^TE\\&=&[\bar{G}^TZ_1\bar{G}+\bar{G}^TZ_2\bar{G}\Delta]+\bar{G}^TE~,\label{UG-error-2}
\end{eqnarray}
where
{\small $$E=\left[\0_{\alpha\times (\ell_{0}-1)}|\eee^T_{\ell_0}|\0_{\alpha\times (\ell_{1}-\ell_{0}-1)}|\cdots|\eee^T_{\ell_{v-1}}|\0_{\alpha\times (n-\ell_{v-1})}\right]~.$$}
\begin{lemma}
\label{lemma}
There are at least $n-k+2$ errors in each of the $\ell_0$th, $\ell_1$th, $\ldots$, $\ell_{v-1}$th columns of $\bar{G}^TY_{\alpha\times n}$.
\end{lemma}
\begin{proof}
From \eqref{UG-error-2}, we have
$$\bar{G}^TY_{\alpha\times n}=P+Q\Delta+\bar{G}^TE.$$
The error vector in $\ell_j$th column is then
\begin{equation}\bar{G}^T\eee^T_{\ell_j}=\left(\eee_{\ell_j}\bar{G}\right)^T~.\label{eG}\end{equation}
Since $\bar{G}$ is a generator matrix of the $[n,k-1]$ RS code, $\eee_{\ell_j}\bar{G}$ in \eqref{eG} is a nonzero codeword in the  RS code. Hence, the number of nonzero symbols in $\eee_{\ell_j}\bar{G}$ is at least $n-k+2$, the minimum Hamming distance of the RS code.
\end{proof}
We next have the main theorem to perform data reconstruction.
\begin{theorem}
\label{thm:main}
Let $\bar{G}^TY_{\alpha\times n}=\tilde{P}+\tilde{Q}\Delta$. Furthermore, let
$\hat{P}$ be the corresponding portion of decoded codeword vector to
$\tilde{P}$ and  $E_P=\hat{P}\oplus \tilde{P}$ be the error pattern vector.
Assume that the data collector accesses all storage nodes and there are $v$,  $1\le
v\le \lfloor\frac{n-k}{2}\rfloor$, of them with errors. Then, there are at
least $n-k+2-v$ nonzero elements in $\ell_{j}$th column of $E_P$,  $0\le j\le
v-1$, and at most $v$ nonzero elements in the rest of the columns of $E_P$.
\end{theorem}
\begin{proof}
Let us focus on the $\ell_{j}$th column of $E_P$. By Lemma~\ref{lemma}, there are at least $n-k+2$ errors in the $\ell_{j}$th column of $\bar{G}^TY_{\alpha\times n}$. $\tilde{P}$ is constructed from $\bar{G}^TY_{\alpha\times n}$ based on \eqref{pq-ij} and \eqref{pq-ji}. If there is only one value of \eqref{pq-ij} and \eqref{pq-ji} that is in error, then the constructed $p_{ij}$ and $q_{ij}$ will be in error. However, when both values are in error, $p_{ij}$ and $q_{ij}$ might accidentally be correct. Among those  $n-k+2$ erroneous positions, there are at least $n-k+2-v$ positions in error after constructing $\tilde{P}$ since at most $v$ errors can be corrected in constructing $\tilde{P}$. It is easy to see that at least $n-k+2-v$ positions are in error that are not among any of the $\ell_0$th, $\ell_1$th, $\ldots$, $\ell_{v-1}$th elements in the $\ell_{j}$th column. These errors are in rows that can be decoded correctly. Hence, there are at least $n-k+2-v$ errors that can be located in $\ell_{j}$th column of $\tilde{P}$ such that there are at least $n-k+2-v$ nonzero elements in the $\ell_{j}$th column of $E_P$. There are at most $v$ rows in $\tilde{P}$ that cannot be decode correctly due to having more than $v$ errors in each of them. Hence, other than those columns with errors in the original matrix $\bar{G}^TY_{\alpha\times n}$, at most $v$ errors will be found in each of the rest of the columns of $\tilde{P}$.
\end{proof}
The above theorem allows us to design a decoding algorithm that can correct up
to $\lfloor\frac{n-k}{2}\rfloor$ errors.\footnote{ In constructing
$\tilde{P}$ we only get $n-1$ values (excluding the diagonal). Since the
minimum Hamming distance of an $[n,k-1]$ RS code is $n-k+2$, the
error-and-erasure decoding can only correct up to
$\lfloor\frac{n-1-k+2-1}{2}\rfloor$ errors.}
In particular, we need to examine the erroneous positions in $\bar{G}^TE$.  Since $1\le v\le
\lfloor\frac{n-k}{2}\rfloor$, we have $n-k+2-v\ge
\lfloor\frac{n-k}{2}\rfloor+1>v$. Thus, the way to locate all erroneous
columns in $\tilde{P}$ is to find out all columns in $E_P$ where the number of
nonzero elements in them are greater than or equal to
$\lfloor\frac{n-k}{2}\rfloor+1$. After we locate all erroneous columns we can
follow a procedure similar to that given in the no error (or single error) case to recover $Z_1$ from
$\hat{P}$.

The above decoding procedure guarantees to recover $Z_1$ ($Z_2$) when all $n$ storage
nodes are accessed. However, it is not very efficient in terms of bandwidth usage.
Next, we present a progressive decoding version of the proposed algorithm that
only accesses enough extra nodes when necessary. Before presenting it, we
need the following corollary.
\begin{corollary}
\label{cor}
Consider that one accesses $k+2v$ storage nodes, among which $v$ nodes are
erroneous and  $1\le v\le \lfloor\frac{n-k}{2}\rfloor$. There are at least
$v+2$ nonzero elements in the $\ell_{J}$th column of $E_P$,  $0\le j\le v-1$, and
at most $v$ among the remaining columns of $E_P$.
\end{corollary}
\begin{proof}
This is a direct result from Theorem~\ref{thm:main} when we delete $n-(k+2v)$ elements in each column of $E_P$ according to the size of $Y_{\alpha\times (k+2v)}$ and $n-k+2-v-\{n-(k+2v)\}=v+2$.
\end{proof}

Based on Corollary~\ref{cor}, we can design a progressive decoding
algorithm~\cite{HAN12} that retrieves extra data from the remaining storage nodes
when necessary. To handle Byzantine fault tolerance, it is necessary to perform
integrity check after the original data is reconstructed.  Two verification
mechanisms have been suggested in~\cite{HAN12-INFOCOM}: cyclic redundancy check
(CRC) and cryptographic hash function. Both mechanisms introduce redundancy to
the original data before they are encoded and are suitable to be used in
combination with the decoding algorithm.

The progressive decoding algorithm starts by accessing $k$ storage nodes.
Error-and-erasure decoding succeeds only when there is no
error. If the integrity check passes, then the data collector recovers the
original data. If the decoding procedure  fails or the integrity check fails,
then the data collector retrieves  two more  blocks of data from the remaining storage
nodes. Since the data collector has $k+2$ blocks of  data, the error-and-erasure
decoding can correctly recover the original data if there is only one erroneous
storage node among the $k+1$ nodes accessed. If the integrity check passes,
then the data collector recovers the original data. If the decoding procedure
fails or the integrity check fails, then the data collector retrieves two more
blocks of  data from the remaining storage nodes. The data collector repeats the same
procedure until it recovers the original data or runs out of the storage nodes.
The detailed decoding procedure is summarized in
Algorithm~\ref{algo:reconstruction-MSR} and its corresponding flowchart is shown in Fig.~\ref{fig:A1}.

Next, we give an example for Algorithm~\ref{algo:reconstruction-MSR} based on a shortened RS code. Let $m = 3$, $n=5$, $k=3$, $\gamma=1$. Then $d = 4$, $\alpha=2$, and
$$ G = \left[\begin{array}{ccccc}
             3 & 5 & 7 & 1 & 0 \\
             2 & 5 & 6 & 0 & 1 \\
             3 & 2 & 4 & 5 & 0 \\
             2 & 2 & 2 & 0 & 2  \end{array} \right]~. $$
Let the information sequence $\m = \left[0~~4~~0~~3~~7~~7\right]$. Then
$$ U = \left[\begin{array}{cccc}
                         0 & 4 & 3 & 7 \\
                         4 & 0 & 7 & 7 \end{array} \right] $$
and  
$$ C = \left[\begin{array}{ccccc}
                         3 & 1 & 7 & 4 & 1 \\
                         0 & 2 & 5 & 2 & 5 \end{array} \right].$$
Assume that the first node is compromised and the vector that the data collector retrieves from the first three nodes for data reconstruction is
$$ Y_{\alpha \times j}  = \left[\begin{array}{ccc}
                         1 & 1 & 7  \\
                         4 & 2 & 5 \end{array} \right]. $$
At the very beginning, we assume that $v=0 \le \lfloor (n-k+1)/2 \rfloor $. 
By Equations~(\ref{UG-no-error}) to (\ref{pq-ji}), we can construct
$$ \tilde{P} = \left[\begin{array}{ccc}
                         0 & 7 & 6  \\
                         7 & 0 & 2  \\
                         6 & 2 & 0 \end{array} \right],
\tilde{Q} = \left[\begin{array}{ccc}
                         0 & 0 & 4  \\
                         0 & 0 & 5  \\
                         4 & 5 & 0 \end{array} \right].$$
We then progressively decode $\tilde{P}$ to obtain
$$ \hat{P} = \left[\begin{array}{ccc}
                         4 & 7 & 6  \\
                         7 & 3 & 2  \\
                         6 & 2 & 0 \end{array} \right]~.$$
Since $v=0$, we can find $\ell_e  = 0$ and  $\ell_c = 3$. Due to  $\ell_e = v$ and $\ell_c = k+v$, we construct
$$ \hat{P}_{\alpha}  =  \left[\begin{array}{cc}
                         4 & 7  \\
                         7 & 3  \end{array} \right]$$
and find 
$$ \bar{G}_{\alpha}  =  \left[\begin{array}{cc}
                         3 & 5  \\
                         2 & 5  \end{array} \right].$$
Finally, $Z_1$ can be recovered and $Z_2$ can be computed similarly as
$$ Z_1  =  \left[\begin{array}{cc}
                         5 & 0  \\
                         0 & 2  \end{array} \right],\ 
Z_2 =  \left[\begin{array}{cc}
                         5 & 5  \\
                         5 & 2  \end{array} \right]~.$$
Therefore, $\tilde{\m} = \left[5~~5~~2~~5~~0~~2\right]$. However, the integrity check of $\tilde{\m}$ fails because the result of the progressive decoding is not correct. The data collector needs to assign $j+2$ and $v+1$ to $j$ and $v$, respectively, and retrieve data from two more nodes. By following the same step as above, we obtain
$$ \tilde{P} = \left[\begin{array}{ccccc}
                         0 & 7 & 6 & 4 & 6  \\
                         7 & 0 & 2 & 2 & 2 \\
                         6 & 2 & 0 & 5 & 1 \\
                         4 & 2 & 5 & 0 & 4 \\
                         6 & 2 & 1 & 4 & 0  \end{array} \right],
\tilde{Q} = \left[\begin{array}{ccccc}
                         0 & 0 & 4 & 5 & 2 \\
                         0 & 0 & 5 & 2 & 0 \\
                         4 & 5 & 0 & 6 & 7 \\
                         5 & 2 & 6 & 0 & 7 \\
                         2 & 0 & 7 & 7 & 0 \end{array} \right],
\hat{P} = \left[\begin{array}{ccccc}
                         7 & 7 & 6 & 4 & 6  \\
                         2 & 0 & 2 & 2 & 2 \\
                         6 & 2 & 0 & 5 & 1 \\
                         3 & 2 & 5 & 0 & 4 \\
                         7 & 2 & 1 & 4 & 0 \end{array} \right]~.$$
Since now $v=1$, we can find $\ell_e  = 1$ and $\ell_c = 4.$ Accordingly, 
$$ \hat{P}_{\alpha}  =  \left[\begin{array}{cc}
                         0 & 2  \\
                         2 & 0  \end{array} \right], \ 
Z_1  =  \left[\begin{array}{cc}
                         0 & 4  \\
                         4 & 0  \end{array} \right],\ 
Z_2 =  \left[\begin{array}{cc}
                         3 & 7  \\
                         7 & 7  \end{array} \right]~.$$
The information sequence is recovered correctly, i.e., $\tilde{\m} = \left[0~~4~~0~~3~~7~~7\right]$.

\begin{algorithm}[h]
\Begin {
$v=0$; $j=k$;\\
The data collector randomly chooses $k$ storage nodes and retrieves encoded data,
$Y_{\alpha\times j}$;\\
\While {$v \le \lfloor\frac{n-k+1}{2}\rfloor$} {
Collect the $j$ columns of $\bar G$ corresponding to accessed storage nodes as  $\bar G_{j}$;\\
Calculate $\bar G_{j}^TY_{\alpha\times j}$;\\
Construct $\tilde{P}$  and  $\tilde{Q}$ by using \eqref{pq-ij} and \eqref{pq-ji};\\
Perform progressive error-and-erasure decoding on each row in $\tilde{P}$  to obtain $\hat{P}$;\\
Locate  erroneous columns in $\hat{P}$  by searching for columns of them with at least $v+2$ errors;
 assume that $\ell_e$ columns found in the previous action;\\
Locate columns in $\hat{P}$ with at most $v$ errors; assume that $\ell_c$ columns found in the previous action;\\
\If{($\ell_e=v$ and $\ell_c=k+v$)} {
Copy  the $\ell_e$ erronous columns of $\hat{P}$ to their corresponding rows to make $\hat{P}$  a symmetric matrix;\\
Collect any $\alpha$ columns in the above $\ell_c$ columns of $\hat{P}$ as $\hat{P}_\alpha$ and find its corresponding $\bar{G}_\alpha$;\\
Multiply  the inverse of $\bar{G}_\alpha$ to $\hat{P}_\alpha$ to recover $\bar{G}_{j}^TZ_1$;\\
Recover $Z_1$ by the inverse of any $\alpha$ rows of $\bar{G}_{j}^T$;\\
Recover $Z_2$ from $\tilde{Q}$ by the same procedure; Recover $\tilde{\m}$ from $Z_1$ and $Z_2$;\\
\If{ integrity-check($\tilde{\m}$) = SUCCESS} {
\Return $\tilde{\m}$;
} }
$j \leftarrow j+2$;\\
Retrieve $2$ more encoded data from remaining storage nodes and merge them into $Y_{\alpha\times j}$; $v\leftarrow v+1$;

}
\Return FAIL;
}
\caption{Decoding of MSR Codes Based on  $(n,k-1)$ RS Code for Data Reconstruction}
\label{algo:reconstruction-MSR}
\end{algorithm}

\begin{figure}
\centering
\includegraphics[width=12cm]{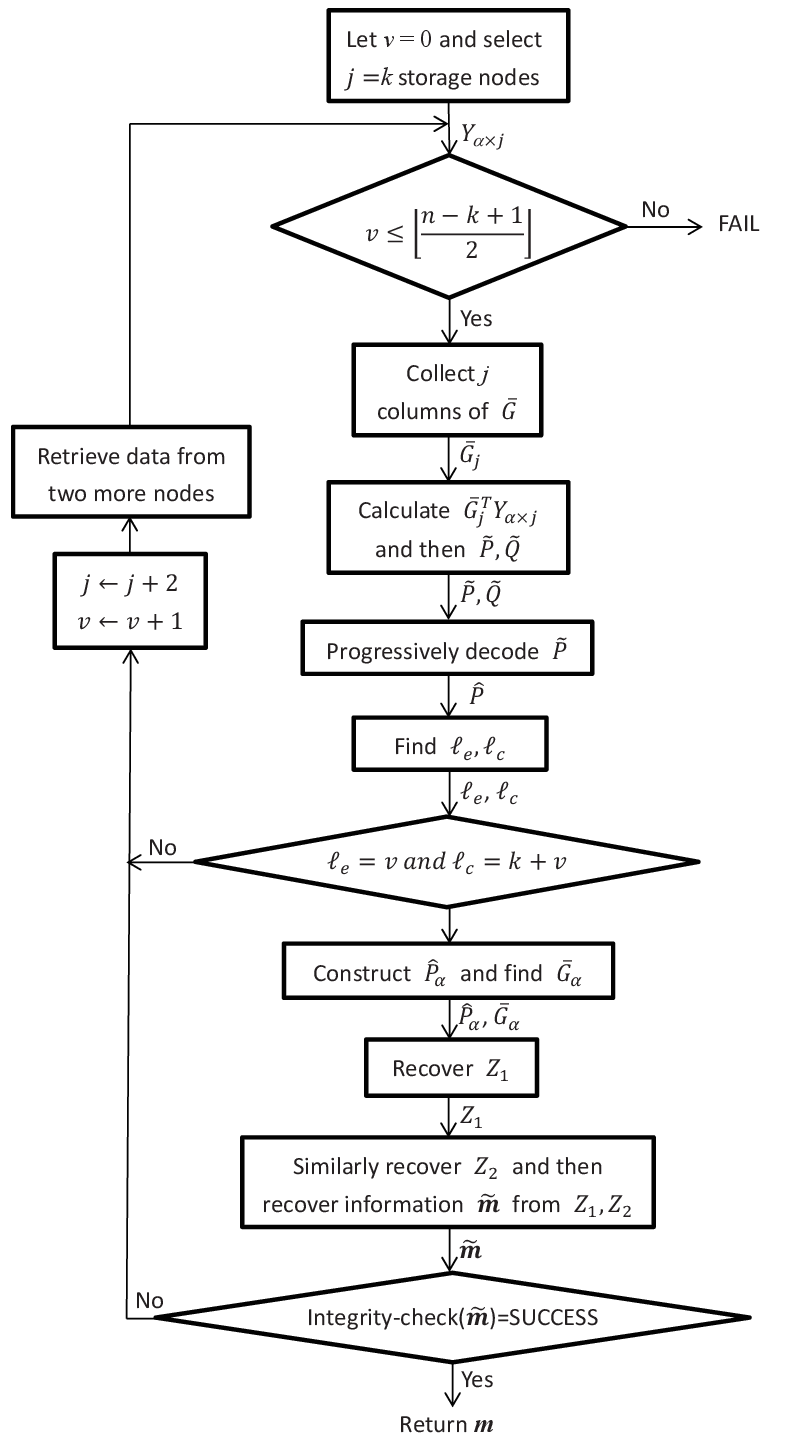}
\caption{Flowchart of Algorithm~\ref{algo:reconstruction-MSR}} \label{fig:A1}
\end{figure}

\section{Encoding and Decoding schemes for Product-Matrix MBR Codes}
\label{SEC:MBR-coding}
In this section, we will find a generator matrix  of the form \eqref{MBR-G-2} such that   the row with the maximum Hamming weight has the least number of nonzero elements.
This generator matrix is thus a least-update-complexity matrix. A decoding scheme for MBR codes that can correct more error patterns is also provided.

\subsection{Encoding Scheme for MBR Codes}
\label{SEC:MBR-encoding}
Let $g(x)=\prod_{j=1}^{n-k}(x-a^j)=\sum_{i=0}^{n-k}g_ix^i$ be the generator polynomial of the $[n,k]$ RS code and $f(x)=\prod_{j=1}^{n-d}(x-a^j)=\sum_{i=0}^{n-d}f_ix^i$ the generator polynomial of the $[n,d]$ RS code, where $a$ is a generator of $GF(2^m)$.\footnote{We assume that $n-k$ and $n-d$ are even.}  A  matrix $G$  can be constructed as

 \begin{eqnarray}
\label{MBR-G-new}
G
&=&\left[\begin{array}{c}
G_k\\
S
\end{array}\right]~,
\end{eqnarray}
where
 \begin{eqnarray}
G_k
&=&\left[\begin{array}{cccccccc}
g_0&g_1&\cdots&g_{n-k}&0&0&\cdots&0\\
0&g_0&\cdots&g_{n-k-1}&g_{n-k}&0&\cdots&0\\
&&&\vdots&&\\
0&\cdots&0&g_{0}&g_1&g_2&\cdots&g_{n-k}\\
\end{array}\right]\label{MBR-G_k}
\end{eqnarray}
and
 \begin{eqnarray}
\label{MBR-B}
S &=&\left[\begin{array}{ccccccccc}
f_0&f_1&\cdots&f_{n-d}&0&0&\cdots&0&0\\
0&f_0&\cdots&f_{n-d-1}&f_{n-d}&0&\cdots&0&0\\
&&&\vdots&&\\
0&\cdots&0&f_{0}&\cdots&f_{n-d}&0&\cdots&0\\
\end{array}\right]~.
\end{eqnarray}
The dimensions of $G_k$ and $S$ are $k\times n$ and $(d-k)\times n$, respectively.
Next, we prove that the main theorem about the rank of $G$ given in~\eqref{MBR-G-new}.
\begin{theorem}
The rank of $G$ given in~\eqref{MBR-G-new} is $d$. That is,  it is a generator matrix of the MBR code.
\end{theorem}
\begin{proof}
Let the codes generated by $G_k$ and $G$ be $\bar{C}$ and $C$, respectively. It can be seen that any row in $G_k$ and $S$ is a cyclic shift of the previous row. Hence, all rows in $G_k$ and $S$ are linearly independent. Now we only consider the linear combination of rows in $G$ chosen from both $G_k$ and $S$. Since $\bar{C}$ is a linear code, the portion of the linear combination that contains only rows from $G_k$ results in a codeword, named $\ccc$, in $\bar{C}$. Assume that the rows chosen from $S$ are the $j_0$th, $j_1$th, $\ldots$, and $j_{\ell-1}$th rows. Recall that $S$ can be represented by a polynomial matrix as
$$B(x)=\left[\begin{array}{c}
f(x)\\
xf(x)\\
x^2f(x)\\
\vdots\\
x^{d-k-1}f(x)
\end{array}\right].$$
Hence, in the polynomial form, the linear combination can be represented as
\begin{equation}
\label{eq:lc}
\ccc(x)+\sum_{i=0}^{\ell-1}b_ix^{j_i-1}f(x)~,
\end{equation}
where $\ccc(x)$ is not the all-zero codeword  and not all $b_i=0$.
Since $c(x)$ is the code polynomial of $\bar {C}$, it is divisible by $g(x)$ and can be represented as $u(x)g(x)$. Assume  that \eqref{eq:lc} is zero. Then we have
\begin{equation}
\label{eq:lc2}
 u(x)g(x)=-f(x)\sum_{i=0}^{\ell-1}b_ix^{j_i-1}~.
\end{equation}
Recall that $g(x)=\prod_{i=1}^{n-k}(x-a^i)$ and $f(x)=\prod_{i=1}^{n-d}(x-a^i)$.
Hence, \begin{equation}
\label{eq:gf}
g(x)=f(x)\prod_{i=n-d+1}^{n-k}(x-a^i)~.
\end{equation}
Substituting \eqref{eq:gf} into \eqref{eq:lc2} we have
\begin{equation}
\label{eq:lc3}
 u(x)\prod_{i=n-d+1}^{n-k}(x-a^i)=-\sum_{i=0}^{\ell-1}b_ix^{j_i-1}~.
\end{equation}
That is, $\sum_{i=0}^{\ell-1}b_ix^{j_i-1}$ is divisible by $\prod_{i=n-d+1}^{n-k}(x-a^i)$. However, the degree of $\prod_{i=n-d+1}^{n-k}(x-a^i)$ is $d-k$ and the degree of $\sum_{i=0}^{\ell-1}b_ix^{j_i-1}$ is at most $d-k-2$ when $\ell=d-k-1$, the largest possible value for $\ell$. Thus,  $\sum_{i=0}^{\ell-1}b_ix^{j_i-1}$ is not divisible by $\prod_{i=n-d+1}^{n-k}(x-a^i)$ since not all $b_i=0$. This is a contradiction.

Since all rows in $G_k$ and $S$ are codewords in $C$, $G$ is then a generator matrix of the $[n,d]$ RS code $C$.
\end{proof}
\begin{corollary}
The $G$ given in~\eqref{MBR-G-new} is the least-update-complexity matrix.
\end{corollary}
\begin{proof}
Since $G_k$ must be the generator matrix of the $[n,k]$ RS code $\bar{C}$, the Hamming weight of each row of $G_k$ is greater than or equal to the minimum Hamming distance of $\bar{C}$, $n-k+1$. Since the degree of $g(x)$ is $n-k$ and itself is a codeword in $\bar{C}$, the nonzero coefficients of $g(x)$ is $n-k+1$ and each row of $G_k$ is with $n-k+1$ Hamming weight. A similar argument can be applied to each row of $S$ such that the Hamming weight of it is $n-d+1$. Thus, the $G$ given in~\eqref{MBR-G-new} has the least number of nonzero elements. Further, Since $G_k$ is the generator matrix of the $[n,k]$ code, the minimum Hamming of its row can have is $n-k+1$, namely, the minimum Hamming distance of the code. Hence, the row with maximum Hamming weight in $G$ is $n-k+1$.
\end{proof}

Since $\bar{C}$ is also a cyclic code, it can be arranged as a systematic code. $G_k$ is then given by

\begin{eqnarray}
 G_k
&=&\left[\begin{array}{cccccccccc}
b_{00}&b_{01}&b_{02}&\cdots&b_{0(n-k-1)}&1&0&0&\cdots&0\\
b_{10}&b_{11}&b_{12}&\cdots&b_{1(n-k-1)}&0&1&0&\cdots&0\\
b_{20}&b_{21}&b_{22}&\cdots&b_{2(n-k-1)}&0&&1&\cdots&0\\
&\vdots&&&&\vdots&&&&\vdots\\
b_{(k-1)0}&b_{(k-1)1}&b_{(k-1)2}&\cdots&b_{(k-1)(n-k-1)}&0&0&0&\cdots&1
\end{array}\right]~,\label{MBR-G-S}\end{eqnarray}
where
$$x^{n-k+i}=u_i(x)g(x)+b_i(x)\mbox{ for } 0\le i\le k-1,$$
and $b_i(x)=b_{i0}+b_{i1}x+\cdots+b_{i(n-k-1)}x^{n-k-1}.$ It is easy to see that $G$ with $G_k$ as a submatrix is still a least-update-complexity  matrix. The advantage of a systematic code will become clear in the decoding procedure of the MBR code.

We now consider the number of encoded symbols that need to be updated while a single data symbol is modified. First, we assume that the modified data symbol is located in  $A_1$. If  the modified data symbol is located in the diagonal of $A_1$,  $(n-k+1)$ encoded symbols need to be updated; otherwise, there are two corresponding encoding symbols in $A_1$ modified such that $2(n-k+1)$ encoded symbols need to be updated. Next, we assume that the modified data symbol is located in  $A_2$. Then $(n-k+1)+(n-d+1)=2n-k-d+2$ encoded symbols need to be updated.

\subsection{Decoding Scheme for MBR Codes}
\label{SEC:MBR-decoding}
The generator polynomial of the RS code encoded by~\eqref{MBR-G-S}
has $a^{n-k},a^{n-k-1},\ldots, a$ as roots. Hence, the progressive
decoding scheme based on the $[n,k]$ RS code given in~\cite{HAN12-INFOCOM} can be applied to decode the
MBR code. The decoding algorithm given in~\cite{HAN12-INFOCOM} is slightly modified as follows.

Assume that the data collector retrieves encoded
symbols from $\ell$ storage nodes $j_0,\ j_1,\ldots,\  j_{\ell-1}$, $k\le \ell\le n$.
The data collector receives $d$ vectors
where each vector has $\ell$ symbols. Denoting the first $k$ vectors among the $d$ vectors as $Y_{k\times \ell}$
and the remaining $d-k$ vectors as $Y_{(d-k)\times \ell}$. By the encoding of the MBR code,  the
codewords in the last $d-k$ rows of $C$ can be viewed as encoded by $G_k$
instead of $G$.  Hence, the decoder of the $[n,k]$ RS code can be applied on
$Y_{(d-k)\times \ell}$ to recover the codewords in the last $d-k$ rows of $C$.

Let
 $\tilde{C}_{(d-k)\times k}$ be the last $k$ columns of the
 codewords recovered by the error-and-erasure decoder in the last $d-k$ rows of $C$. Since the code generated by ~\eqref{MBR-G-S} is a systematic code, $A_2$ in $U$ can  be
reconstructed as
\begin{eqnarray}
\tilde{A}_2=\tilde{C}_{(d-k)\times k}~.\label{A_2}
\end{eqnarray}
We then calculate  the $j_0$th, $j_1$th,
$\ldots$, $j_{\ell-1}$th columns of $\tilde{A}_2^T\cdot B$ as $E_{k\times\ell}$, and  subtract $E_{k\times\ell}$
from $Y_{k\times \ell}$:
\begin{eqnarray}
Y'_{k\times\ell}=Y_{k\times \ell}-E_{k\times\ell}~.\label{Y_k}
\end{eqnarray}
Applying the error-and-erasure decoding algorithm of the $[n,k]$ RS code again on $Y'_{k\times\ell}$ we can reconstruct $A_1$ as
\begin{eqnarray}
\tilde{A}_1=\tilde{C}_{k\times k}~.\label{A_1}
\end{eqnarray}

 The decoded information sequence is then verified by data integrity check. If the integrity check is
passed, the data reconstruction is successful; otherwise the progressive
decoding procedure is applied, where two more storage nodes need to be accessed
from the remaining storage nodes in each round until no further errors are
detected.

The decoding capability of the above decoding algorithm is $\frac{n-k}{2}$. Since each erroneous storage node sends $\alpha=d$ symbols to the data collector, in general, not all $\alpha$ symbols are wrong if failures in the storage nodes are caused by random faults. Hence, the decoding algorithm given in~\cite{HAN12-INFOCOM} can be modified as follows to extend error correction capability.  After decoding $Y_{(d-k)\times \ell}$, one can locate the erroneous columns of $Y_{(d-k)\times \ell}$ by comparing the decoded result to it. Assume that there are $v$ erroneous columns located. Delete the corresponding columns in $E_{k\times\ell}$ and $Y_{k\times \ell}$ and we have
\begin{equation}Y'_{k\times(\ell-v)}=Y_{k\times (\ell-v)}-E_{k\times (\ell-v)}~.\label{Y_v}\end{equation}
Applying the error-and-erasure decoding algorithm of the $[n,k]$ RS code again on $Y'_{k\times(\ell-v)}$ to reconstruct $A_1$ if $\ell-v\ge k$; otherwise the progressive decoding is applied.
The modified decoding algorithm  is summarized in Algorithm~\ref{algo:reconstruction-MBR} and its corresponding flow chart is shown in Fig.~\ref{fig:A2}. The advantage of the modified decoding algorithm is that it can correct errors up to $$\frac{n-k}{2}+ \left\lfloor\frac{n-k+1-\lfloor\frac{n-k+1}{2}\rfloor}{2}\right\rfloor$$
even though not all error patterns up to such number of errors can be corrected.

Next, we  give an example for Algorithm~\ref{algo:reconstruction-MBR} based on a shortened RS code. Let $m = 3$, $n=5$, $k=3$, $d = 4$. Then $\alpha=4$ and
$$ G = \left[\begin{array}{ccccc}
             3 & 6 & 1 & 0 & 0 \\
             1 & 1 & 0 & 1 & 0 \\
             3 & 7 & 0 & 0 & 1 \\
             2 & 1 & 0 & 0 & 0  \end{array} \right]. $$
Let the information sequence $\m = \left[0~~4~~0~~3~~7~~0~~3~~7~~7\right]$. Then 
$$ U = \left[\begin{array}{cccc}
                         0 & 4 & 0 & 3 \\
                         4 & 3 & 7 & 7 \\
                         0 & 7 & 0 & 7 \\
                         3 & 7 & 7 & 0 \end{array} \right]$$
and 
$$ C  = \left[\begin{array}{ccccc}
                         2 & 7 & 0 & 4 & 0 \\
                         3 & 2 & 4 & 3 & 7 \\
                         2 & 0 & 0 & 7 & 0 \\
                         0 & 5 & 3 & 7 & 7 \end{array} \right]~. $$
 Assume that the first node is compromised and the vector that the data collector retrieves from the first three nodes for data reconstruction is 
$$ Y_{d \times \ell}  = \left[\begin{array}{ccc}
                         1 & 7 & 0  \\
                         0 & 2 & 4  \\
                         4 & 0 & 0  \\
                         6 & 5 & 3 \end{array} \right].$$
At the beginning, $\ell = k$ and we assume that $v=0$. We decode the last $d-k$ rows of $Y_{d \times \ell}$ and obtain
$$ \bar{C}_{(d-k)\times k} = \left[\begin{array}{ccc}
                         3 & 6 & 3 \end{array} \right]~.$$
By Equations~(\ref{A_2}) to (\ref{A_1}),
$$ \tilde{A}_2 = \left[\begin{array}{ccc}
                         3 & 6 & 3 \end{array} \right],\
Y^{\prime}_{k \times (\ell-v)} = \left[\begin{array}{ccc}
                         7 & 4 & 0  \\
                         7 & 4 & 4  \\
                         2 & 3 & 0 \end{array} \right],\
\tilde{A}_1 = \left[\begin{array}{ccc}
                         0 & 1 & 2  \\
                         4 & 2 & 7  \\
                         0 & 0 & 7 \end{array} \right]~.$$
Therefore, $\tilde{\m} = \left[0~~1~~2~~2~~7~~7~~3~~6~~3\right].$ The integrity check of $\m$ also fails. The data collector needs to retrieve data from two more nodes and assign $\ell + 2$ to $\ell$. By following the same step as above, $\bar{C}_{(d-k)\times k} = \tilde{A}_2 = \left[\begin{array}{ccc}
                         3 & 7 & 7 \end{array} \right]$, 
$$Y^{\prime}_{k \times (\ell-v)} = \left[\begin{array}{cccc}
                         4 & 0 & 4 & 0 \\
                         5 & 4 & 3 & 7 \\
                         7 & 0 & 7 & 0 \end{array} \right],
\tilde{A}_1 = \left[\begin{array}{ccc}
                         0 & 4 & 0  \\
                         4 & 3 & 7  \\
                         0 & 7 & 0 \end{array} \right].$$
The information sequence is recovered correctly, i.e., $\tilde{\m} = \left[0~~4~~0~~3~~7~~0~~3~~7~~7\right]$.

\begin{algorithm}[h]
\Begin {
The data collector randomly chooses $k$ storage nodes and retrieves encoded data,
$Y_{d \times k}$;\\

$\ell \leftarrow k$;\\
\Repeat{$\ell \ge n-2$} {
Perform progressive error-erasure decoding on last $d-k$ rows in $Y_{d\times\ell}$, $Y_{(d-k)\times \ell}$, to recover $\tilde{C}$ (error-erasure decoding performs $d-k$ times);\\
Locate the erroneous columns in $Y_{(d-k)\times \ell}$ (assume to have $v$ columns);\\
Calculate $\tilde{A}_2$ via~\eqref{A_2};\\
Calculate $\tilde{A}_2\cdot B$  and obtain $Y'_{k\times (\ell-v)}$ via~\eqref{Y_v};\\
\If{($\ell-v\ge k$)} {Perform progressive error-erasure decoding  on $Y'_{k\times (\ell-v)}$ to recover the first $k$ rows in codeword vector (error-erasure decoding performs $k$ times);\\
Calculate $\tilde{A}_1$ via~\eqref{A_1};\\
Recover the information sequence $\tilde{\m}$ from $\tilde{A}_1$ and $\tilde{A}_2$;\\
 \If{integrity-check($\tilde{\m}$) = SUCCESS} {
\Return $\tilde{\m}$;
} }
$\ell\leftarrow \ell+2$;\\
Retrieve two more encoded data from remaining storage nodes  and merge them into $Y_{d\times \ell}$; \\
}
\Return FAIL;
}
\caption{Decoding of MBR Codes for Data Reconstruction}
\label{algo:reconstruction-MBR}
\end{algorithm}

\begin{figure}
\centering
\includegraphics[width=12cm]{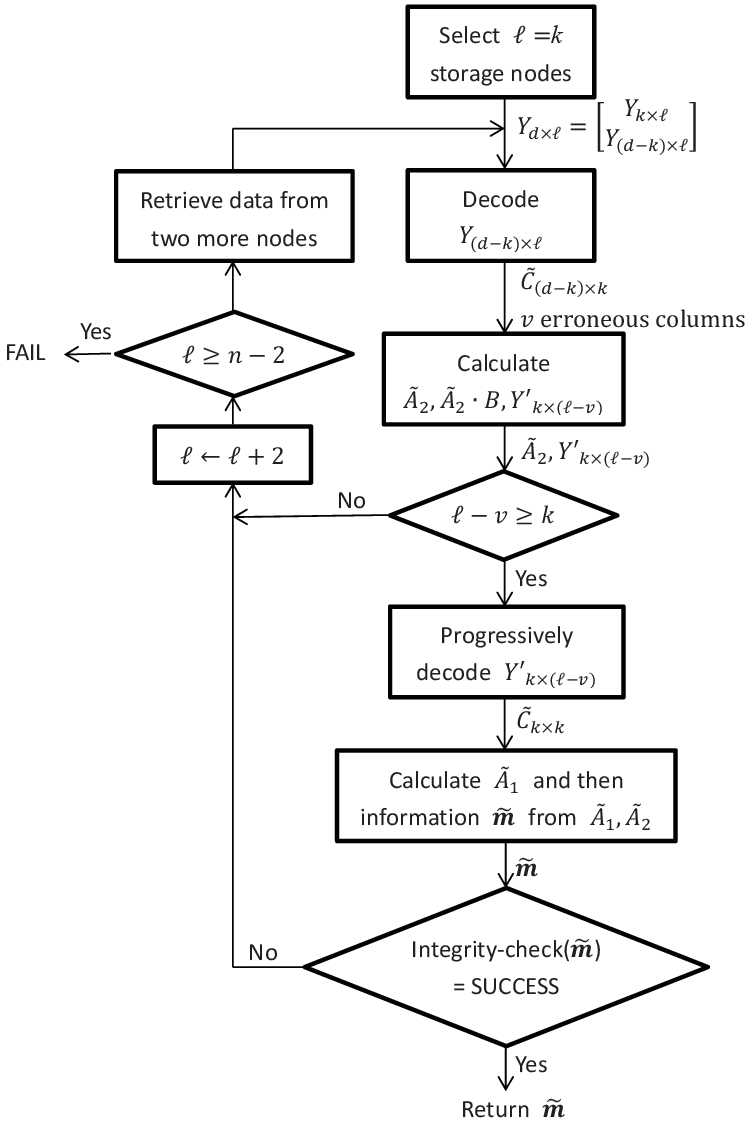}
\caption{Flow chart of Algorithm~\ref{algo:reconstruction-MBR}} \label{fig:A2}
\end{figure}

One important function of regenerating codes is to perform data regeneration with least repair bandwidth while one node is failed. Since the decoding schemes proposed in~\cite{HAN12-INFOCOM} can be applied directly without modification to the proposed MSR and MBR codes in this work, the decoding schemes of data regeneration for these codes are omitted in this work. The interested readers can refer to \cite{HAN12-INFOCOM} for details on these decoding schemes.

\section{Performance Evaluation}
\label{SEC:eval}

\begin{figure}
\centering
\includegraphics[width=8cm]{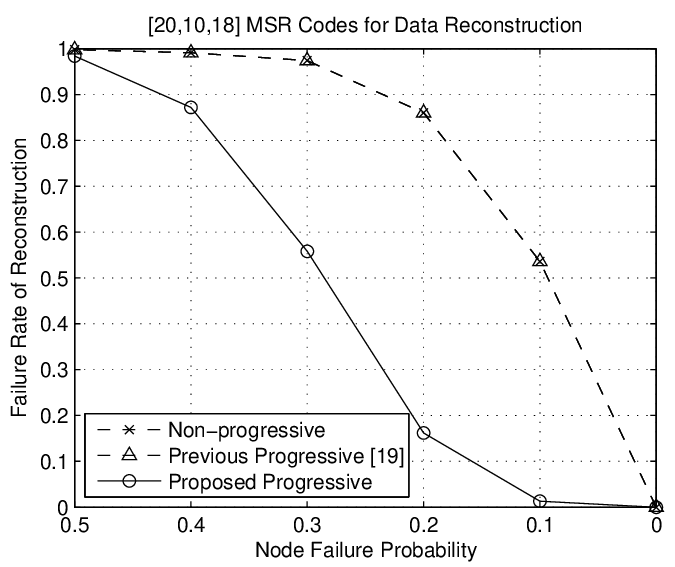}
\caption{Comparison of the failure rate between the algorithm in~\cite{HAN12-INFOCOM} and the proposed algorithm for $[20,10,18]$ MSR codes} \label{fig:fig1}
\end{figure}

\begin{figure}
\centering
\includegraphics[width=8cm]{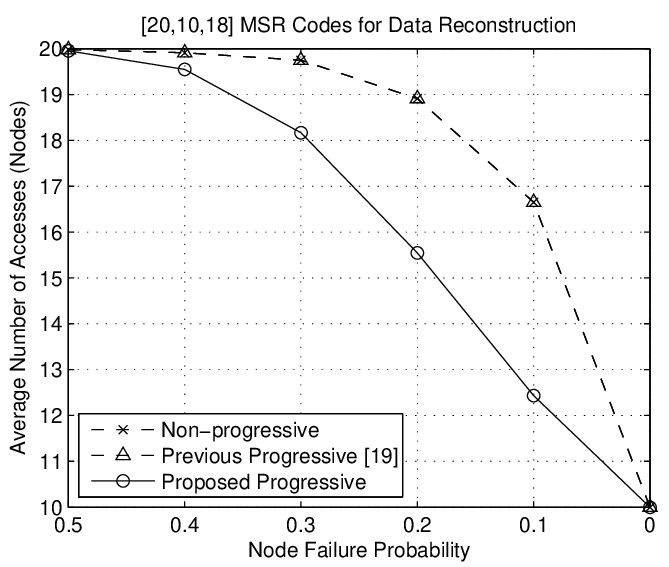}
\caption{Comparison of the number of node accesses between the algorithm in~\cite{HAN12-INFOCOM} and the proposed algorithm for $[20,10,18]$ MSR codes} \label{fig:fig2}
\end{figure}

In this section, we first analyze the fault-tolerance capability of the proposed
codes in the presence of crash-stop and Byzantine failures, security strength with malicious attack, and then carry out
numerical simulations to evaluate the performance  for proposed schemes.

The  fault-tolerance capability of product-matrix MSR and MBR codes has been investigated fully in~\cite{HAN12-INFOCOM} where CRC or cryptographic hash function is adopted as the data integrity check. Their error-correction capability  was also presented in~\cite{RAS12}.

We need to verify whether the
reconstructed data are correct. Progressive decoding algorithms are
implemented that  incrementally retrieve additional stored data and perform
data reconstruction  when errors have been detected. Since cryptographic hash function has better security strength than CRC on data integrity check, it is  adopted to verify the integrity of stored data. In particular, for data reconstruction,  the hash value is coded along with the original data and distributed among storage nodes.

We first consider two types of failures,
 crash-stop failures and Byzantine failures. Nodes are assumed to fail
independently. In both cases, the
fault-tolerance capability is measured by the maximum number of failures that the
system can handle to maintain functionality.

A crash-stop failure on a node can be viewed as an erasure in the codeword. Since
$k$ nodes need to be alive for data reconstruction, the maximum number of crash-stop failures that can be
tolerated in data reconstruction is $n-k$. Note that since all accessed nodes
contain correct data, the associated hash values are also correct.

For an error-correcting code, two additional correct code fragments are needed to
correct one erroneous code fragment. Thus, with the proposed MSR decoding algorithm,
$\lfloor \frac{n-k}{2}\rfloor$ erroneous
nodes can be tolerated in data reconstruction. For the proposed MBR decoding algorithm, not only any $\frac{n-k}{2}$ erroneous
nodes can be tolerated but  it can also correct errors up to $$\frac{n-k}{2}+ \left\lfloor\frac{n-k+1-\lfloor\frac{n-k+1}{2}\rfloor}{2}\right\rfloor$$
even though not all error patterns up to such number of errors can be corrected.

In analyzing the security strength with malicious attacks, we consider forgery attacks, where Byzantine attackers try to disrupt the
data reconstruction process by forging data collaboratively.
In other words, collusion among compromised nodes is considered. We want to determine
the minimum number of  compromised nodes to forge the data in data reconstruction.
By using cryptographic hash functions, the security strength can be increased since the operation to obtain the hash value is non-linear. In this case, the attacker  needs to obtain the original information data to forge the hash value. Hence, the attacker needs to compromise at least $k$ nodes  in data reconstruction.

The proposed data reconstruction algorithms for MSR and MBR codes have also been evaluated by Monte Carlo simulations. From now on, the codes based on shortened RS codes are employed for simulations.  They are compared with the data reconstruction algorithms previously proposed in~~\cite{HAN12-INFOCOM}. The performance of a traditional decoding scheme that is non-progressive is also provided for comparison purposes.\footnote{Since no data integrity check is performed in the deocding algorithms given in~\cite{RAS12}, to reach error-correction capability of the MSR and MBR codes, $n$ nodes need to be accessed. Hence, the number of accessed nodes in deocding algorithms in~\cite{RAS12} are much larger than  those of the non-progressive version presented here.} After $k$ nodes are accessed, if the integrity check fails, the data collector will access all remaining $n-k$ nodes in data reconstruction in the non-progressive decoding scheme. Each
data point is generated from $10^3$ simulation runs. Storage nodes may fail
arbitrarily with the Byzantine failure probability ranging from $0$ to $0.5$. In both
schemes, $[n,k,d]$ and $m$ are chosen to be $[20,10,18]$ and $5$, respectively.

In the first set of simulations, we compare the proposed algorithm with the progressive algorithm in ~\cite{HAN12-INFOCOM} and the non-progressive algorithm in terms of the failure rate of  reconstruction and the average number of node accesses, which indicates the required bandwidth for data reconstruction. Failure rate is defined as the percentage of runs for which reconstruction fails (due to insufficient number of healthy storage nodes). Figure~\ref{fig:fig1} shows that the proposed algorithm can successfully reconstruct the data with much higher probability than the previous progressive or non-progressive algorithm for the same node failure probability. For example, when the node failure probability is $0.1$, only about 1\% of the time, reconstruction fails using the proposed algorithm, in contrast to 50\% with the old algorithm. The advantage of the proposed algorithm is also pronounced in the average number of accessed nodes for data reconstruction, as illustrated in Fig.~\ref{fig:fig2}. For example, on an average, only $2.5$ extra nodes are needed by the proposed algorithm under the node failure probability of $0.1$; while over $6.5$ extra nodes are required by the old algorithm  in ~\cite{HAN12-INFOCOM}. It should be noted that the actual saving attained by the new algorithm depends on the setting of $n$, $k$, $d$ and the number of errors.

\begin{figure}
\centering
\includegraphics[width=8cm]{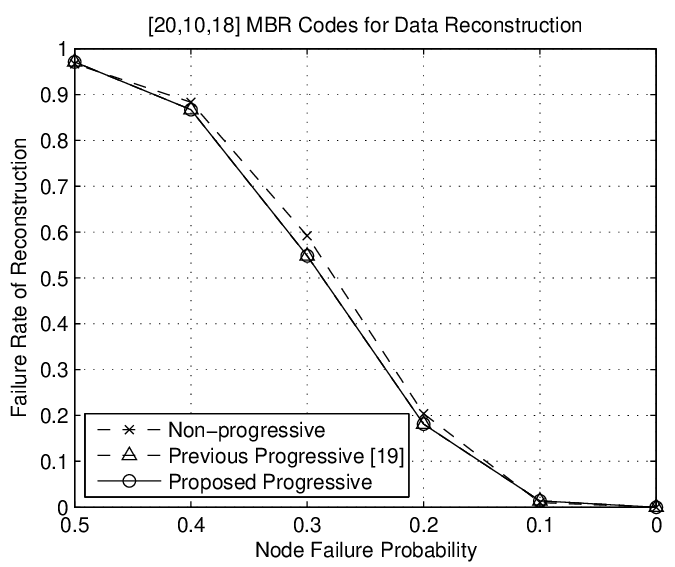}
\caption{Failure-rate comparison between the previous algorithm in~\cite{HAN12-INFOCOM} and the proposed algorithm for $[20,10,18]$ MBR codes} \label{fig:fig3}
\end{figure}

\begin{figure}
\centering
\includegraphics[width=8cm]{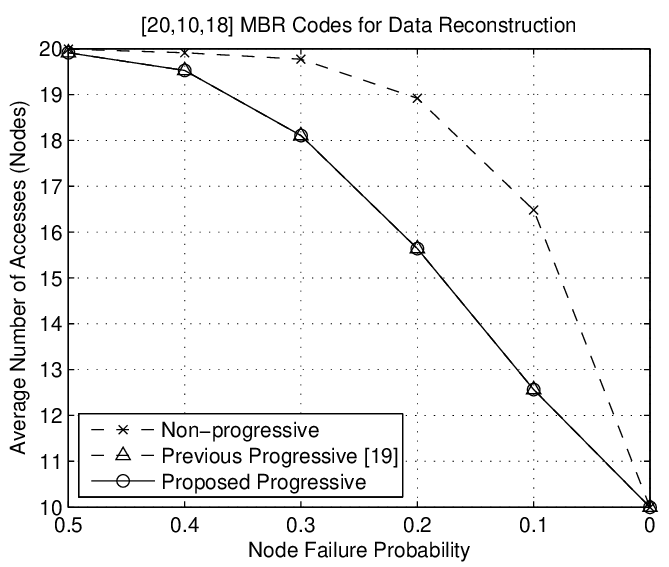}
\caption{Node-access comparison between the previous algorithm in~\cite{HAN12-INFOCOM} and the proposed algorithm for $[20,10,18]$ MBR codes} \label{fig:fig4}
\end{figure}

The previous and proposed decoding algorithms for MBR codes are compared in the second set of simulations. Figures~\ref{fig:fig3} and~\ref{fig:fig4} show that both of the progressive algorithms have identical failure rates of reconstruction and average number of accessed nodes. This result implies that the specific error patterns, which only the proposed algorithm is able to handle for successful data reconstruction, do not happen very frequently. However, the computational complexity of the proposed algorithm for MBR encoding is much lower since no matrix inversion and multiplications are needed in (\ref{A_2}) and (\ref{A_1}). Moreover, both the progressive algorithms are better than the non-progressive algorithm in failure rates of reconstruction and average number of accessed nodes.

In the evaluation of the update complexity, two measures are considered: the metric given in~\cite{RAW11} and the number of updated symbols when a single data symbol is modified. The first metric corresponds to   the maximum number of nonzero elements in all rows of the generator matrix $G$. Denote by $\eta(R)$ the ratio of the update complexity of the proposed generator matrix to that of the generator matrix given in~\cite{RAS11}, where $R=k/n$. It can be seen that,
$$\eta_{MSR}(R)=\frac{n-\alpha+1}{n}\approx 1-R$$
for MSR codes since the generator matrix of the MSR code proposed in~\cite{RAS11} is a Vandermonde matrix. Two types of generator matrices of the MBR codes have been proposed in~\cite{RAS11}: the Vandermonde matrix  and  a systematic matrix  based on Cauchy matrix. With Vandermonde matrix,
$$\eta_{MBR}(R)=\frac{n-k+1}{n}\approx 1-R~.$$
The systematic matrix based on Cauchy matrix is given by~\cite{RAS11}
$$\left[\begin{array}{cc}
I_k&\phi^T\\
\0&\Delta^T
\end{array}\right]~,$$
where $I_k$ is the $k\times k$ identity matrix, $\0$ is the $(d-k)\times k$ all-zero matrix, and $[\phi\ \Delta]$ is a Cauchy matrix. Since all elements in the Cauchy matrix are nonzero,
$$\eta_{MBR}(R)=\frac{n-k+1}{n-k+1}=1~.$$

 \begin{table*}[tbh]
\caption{Comparison on the average number of updated symbols while a single data symbol is modified}
\label{tab:evaluation}
\begin{minipage}{8cm}
 \begin{center}
\begin{tabular}{|c||c|c||c|c||}\hline
&\mc{2}{c||}{MSR code}&\mc{2}{c||}{MBR code}\\ \hline\hline
&[20 10 18]&[100 40 78] &[20 10 18]&[100 40 78] \\ \hline
 Proposed method&22 &121 & 8&48\\ \hline
 Vandermonde matrix&36 &195 & 19&99\\ \hline
 Systematic version based on linear remapping~\cite{RAS11}*&88&2323&34&807\\ \hline
 Systematic version based on Cauchy matrix~\cite{RAS11}&-&-&10&60\\ \hline
\end{tabular}
\end{center}
* The numbers are obtained from simulation results
\end{minipage}
\end{table*}

The number of updated symbols that need to be modified when a single data symbol is changed in MSR and MBR codes are summarized in Table~\ref{tab:evaluation}. By the arguments given in previous sections, the average number of updated symbols when a single data symbol is modified for the proposed MSR and MBR codes are $2(n-\alpha+1)\frac{\alpha}{\alpha+1}$ and
$\frac{kd(n-k+1)+k(d-k)(n-d+1)}{2kd-k(k-1)}$, respectively. These numbers for Vandermonde-matrix based MSR and MBR codes are $2n\frac{\alpha}{\alpha+1}$ and
$\frac{n(2kd-k^2)}{2kd-k(k-1)}$, respectively. The number is $\frac{kd(n-k+1)+k(d-k)(n-k)}{2kd-k(k-1)}$  for the systematic MBR code based on Cauchy matrix. Note that, the numbers for systematic codes based on linear remapping are obtained from simulations. From Table~\ref{tab:evaluation}, one can observe that the proposed method has the best performance on the number of updated symbols when a single data symbol is modified, and the systematic version based on linear remapping performs the worst among all schemes in the table. For example, for the $[20,10,18]$ MSR code, the average number
of encoded symbols that need to be updated for a single data symbol modification is
$88$ in the systematic version based on linear remapping but only $22$ with the
proposed encoding matrix. This is a 4-fold improvement in complexity. In the case of the $[100,40,78]$ MSR code, the improvement is 19-fold. Hence, the proposed approach has much lower update complexity than the
systematic approach. It can be seen that after linear remapping, the modified symbols almost occur in all check positions of the code vector. This is because even when only one data symbol is modified, due to the symmetry requirement on the information matrix, the modification propagates to check positions of all codewords (rows) in the code vector through linear remapping. One can also observe that even though the Cauchy-based MBR code results in the same maximum number of nonzero elements in all rows of the generator matrix  as the proposed MBR code, it requires more symbol updates when a single data symbol is modified.

\section{Conclusion}
\label{SEC:conclude}
In this work, we proposed new encoding and decoding schemes for the $[n,d]$ error-correcting MSR and MBR codes that generalize the previously proposed codes in~\cite{HAN12-INFOCOM}. Through both theoretical analysis and numerical simulations, we demonstrated the superior  error correction capability, low update complexity and low computation complexity of the new codes.

Clearly, there is a trade-off between the update complexity and error correction capability of regenerating codes. In this work, we  found encoders of  product-matrix regenerating codes and then optimized their update complexity. Possible future work includes the study of encoding schemes that first design regenerating codes with good  update complexity and then optimize their error correction capability.

The least update-complexity codes  in this work minimize the  maximum number of nonzero elements in all rows of the generation matrix, but they do not minimize the number of symbol updates when a single data symbol is modified.  For instance, due to symmetry requirement on the information vector,  two symbols need to be updated in the information vector during the encoding process for a single modified symbol in some cases.  Another possible future work is to seek codes with the least number of updated encoded symbols.

\begin{thebibliography}{10}
\providecommand{\url}[1]{#1}
\csname url@samestyle\endcsname
\providecommand{\newblock}{\relax}
\providecommand{\bibinfo}[2]{#2}
\providecommand{\BIBentrySTDinterwordspacing}{\spaceskip=0pt\relax}
\providecommand{\BIBentryALTinterwordstretchfactor}{4}
\providecommand{\BIBentryALTinterwordspacing}{\spaceskip=\fontdimen2\font plus
\BIBentryALTinterwordstretchfactor\fontdimen3\font minus
  \fontdimen4\font\relax}
\providecommand{\BIBforeignlanguage}[2]{{%
\expandafter\ifx\csname l@#1\endcsname\relax
\typeout{** WARNING: IEEEtran.bst: No hyphenation pattern has been}%
\typeout{** loaded for the language `#1'. Using the pattern for}%
\typeout{** the default language instead.}%
\else
\language=\csname l@#1\endcsname
\fi
#2}}
\providecommand{\BIBdecl}{\relax}
\BIBdecl

\bibitem{GHE03}
S.~Ghemawat, H.~Gobioff, and S.-T. Leung, ``The \protect{Google} file system,''
  in \emph{Proc. of the 19th ACM SIGOPS Symp. on Operating Systems Principles},
  Bolton Landing, NY, October 2003.

\bibitem{KUB00}
J.~K. et~al., ``\protect{OceanStore:} an architecture for global-scale
  persistent storage,'' in \emph{Proc. of the 9th International Conference on
  Architectural Support for programming Languages and Operating Systems},
  Cambridge, MA, November 2000.

\bibitem{BHA04}
R.~Bhagwan, K.~Tati, Y.~Cheng, S.~Savage, and G.~Voelker, ``Total recall:
  system support for automated availability management,'' in \emph{Proc. of the
  1st Conf. on Networked Systems Design and Implementation}, San Francisco, CA,
  March 2004.

\bibitem{DIM07}
A.~G. Dimakis, P.~B. Godfrey, M.~Wainwright, and K.~Ramchandran, ``Network
  coding for distributed storage systems,'' in \emph{Proc. of 26th IEEE
  International Conference on Computer Communications (INFOCOM)}, Anchorage,
  Alaska, May 2007, pp. 2000--2008.

\bibitem{DIM10}
A.~G. Dimakis, P.~B. Godfrey, Y.~Wu, M.~Wainwright, and K.~Ramchandran,
  ``Network coding for distributed storage systems,'' \emph{IEEE Trans. Inform.
  Theory}, vol.~56, pp. 4539 -- 4551, September 2010.

\bibitem{WU07}
Y.~Wu, A.~G. Dimakis, and K.~Ramchandran, ``Deterministic regenerating codes
  for distributed storage,'' in \emph{Proc. of 45th Annual Allerton Conference
  on Control, Computing, and Communication}, Urbana-Champaign, Illinois,
  September 2007.

\bibitem{WU10}
Y.~Wu, ``Existence and construction of capacity-achieving network codes for
  distributed storage,'' \emph{IEEE Journal on Selected Areas in
  Communications}, vol.~28, pp. 277 -- 288, February 2010.

\bibitem{CUL09}
D.~F. Cullina, ``Searching for minimum storage regenerating codes,'' California
  Institute of Technology Senior Thesis, 2009.

\bibitem{WU09}
Y.~Wu and A.~G. Dimakis, ``Reducing repair traffic for erasure coding-based
  storage via interference alignment,'' in \emph{Proc. IEEE International
  Symposium on Information Theory}, Seoul, Korea, July 2009, pp. 2276--2280.

\bibitem{RAS09}
K.~V. Rashmi, N.~B. Shah, P.~V. Kumar, and K.~Ramchandran, ``Explicit
  construction of optimal exact regenerating codes for distributed storage,''
  in \emph{Proc. of 47th Annual Allerton Conference on Control, Computing, and
  Communication}, Urbana-Champaign, Illinois, September 2009, pp. 1243--1249.

\bibitem{PAW11}
S.~Pawar, S.~El~Rouayheb, and K.~Ramchandran, ``Securing dynamic distributed
  storage systems against eavesdropping and adversarial attacks,''
  \emph{Information Theory, IEEE Transactions on}, vol.~57, no.~10, pp.
  6734--6753, 2011.

\bibitem{OGG11}
F.~Oggier and A.~Datta, ``Byzantine fault tolerance of regenerating codes,'' in
  \emph{Peer-to-Peer Computing (P2P), 2011 IEEE International Conference on},
  2011, pp. 112--121.

\bibitem{RAS11}
K.~V. Rashmi, N.~B. Shah, and P.~V. Kumar, ``Optimal exact-regenerating codes
  for distributed storage at the \protect{MSR and MBR} points via a
  product-matrix construction,'' \emph{IEEE Trans. Inform. Theory}, vol.~57,
  pp. 5227--5239, August 2011.

\bibitem{SHA12}
N.~Shah, K.~V. Rashmi, P.~Kumar, and K.~Ramchandran, ``Interference alignment
  in regenerating codes for distributed storage: Necessity and code
  constructions,'' \emph{Information Theory, IEEE Transactions on}, vol.~58,
  no.~4, pp. 2134--2158, 2012.

\bibitem{CAD10}
H.~M. V.~R.~Cadambe, S. A.~Jafar, ``Distributed data storage with minimum
  storage regenerating codes - exact and functional repair are asymptotically
  equally efficient,'' arXiv:1004.4299v1 [cs.IT] 24 Apr 2010.

\bibitem{SUH11}
C.~Suh and K.~Ramchandran, ``Exact-repair mds code construction using
  interference alignment,'' \emph{IT}, vol.~57, pp. 1425 -- 1442, March 2011.

\bibitem{CAD08}
V.~R. Cadambe and C.~Jafar, ``Interference alignment and degrees of freedom of
  the k-user interference channel,'' \emph{IT}, vol.~54, pp. 3425 -- 3441,
  August 2008.

\bibitem{MAD08}
M.~A. Maddah-Ali, A.~S. Motahari, and A.~K. Khandani, ``Communication over
  \protect{MIMO X} channels: Interference alignment, decomposition, and
  performance analysis,'' \emph{IT}, vol.~54, pp. 3457 -- 3470, August 2008.

\bibitem{HAN12-INFOCOM}
Y.~S. Han, R.~Zheng, and W.~H. Mow, ``Exact regenerating codes for byzantine
  fault tolerance in distributed storage,'' in \emph{Proc. of the IEEE INFOCOM
  2012}, Orlendo, FL, March 2012.

\bibitem{RAS12}
K.~Rashmi, N.~Shah, K.~Ramchandran, and P.~Kumar, ``Regenerating codes for
  errors and erasures in distributed storage,'' in \emph{Proc. of the 2012 IEEE
  International Symposium on Information Theory}, Cambridge, MA, July 2012.

\bibitem{RAW11}
A.~S. Rawat, S.~Vishwanath, A.~Bhowmick, and E.~Soljanin, ``Update efficient
  codes for distributed storage,'' in \emph{Proc. of the 2011 IEEE
  International Symposium on Information Theory}, Saint Petersburg, Russia,
  July 2011.

\bibitem{LIN04}
S.~Lin and D.~J. Costello, Jr., \emph{Error Control Coding: Fundamentals and
  Applications}, 2nd~ed.\hskip 1em plus 0.5em minus 0.4em\relax Englewood
  Cliffs, NJ: Prentice-Hall, Inc., 2004.

\bibitem{MOO05}
T.~K. Moon, \emph{Error Correction Coding: Mathematical Methods and
  Algorithms}.\hskip 1em plus 0.5em minus 0.4em\relax Hoboken, NJ: John Wiley
  \& Sons, Inc., 2005.

\bibitem{MAC77}
F.~J. MacWilliams and N.~J.~A. Sloane, \emph{The Theory of Error-Correcting
  Codes}.\hskip 1em plus 0.5em minus 0.4em\relax New York, NY: Elsevier Science
  Publishing Company, Inc., 1977.

\bibitem{HAN12}
Y.~S. Han, S.~Omiwade, and R.~Zheng, ``Progressive data retrieval for
  distributed networked storage,'' \emph{IEEE Trans. on Parallel and
  Distributed Systems}, vol.~23, pp. 2303--2314, December 2012.

\end{thebibliography}

\end{document}